\documentclass[11pt]{article}
\usepackage{microtype}
\usepackage{fullpage,amsthm,amssymb,amsmath}
\usepackage{graphicx,algorithmic,algorithm}
\usepackage{enumerate}
\usepackage{tikz}
\usetikzlibrary{shapes}
\usepackage{hyperref}
\usepackage{complexity}
\usepackage[normalem]{ulem}
\usepackage{caption, subcaption}
\usepackage{todonotes}

\usepackage{xcolor}
\definecolor{dark-red}{rgb}{0.4,0.15,0.15}
\definecolor{dark-blue}{rgb}{0.15,0.15,0.4}
\definecolor{medium-blue}{rgb}{0,0,0.5}
\definecolor{gray}{rgb}{0.5,0.5,0.5}

\hypersetup{
    pdftitle=   {A unified polynomial-time algorithm for Feedback Vertex Set on graphs of bounded mim-width},
   pdfauthor=  {Lars Jaffke, O-joung Kwon and Jan Arne Telle}, 
    colorlinks, linkcolor={dark-red},
    citecolor={dark-blue}, urlcolor={medium-blue}
}

\usepackage{authblk}
\newtheorem{theorem}{Theorem}
\newtheorem{lemma}[theorem]{Lemma}
\newtheorem{proposition}[theorem]{Proposition}
\newtheorem{corollary}[theorem]{Corollary}

\newtheorem*{thm:main}{Theorem~\ref{thm:main}}

\newcommand\abs[1]{\lvert #1\rvert}
\newcommand\card[1]{\lvert #1 \rvert}

\newcommand\dist{\operatorname{dist}}

\newcommand\mimw{\operatorname{mimw}}
\newcommand\mimval{\operatorname{mim}}
\newcommand\mim{\operatorname{mim}}

\newcommand\bd{\operatorname{bd}}

\newcommand\fvs{\textsc{Feedback Vertex Set}}
\newcommand\dptable{\mathcal{T}}

\newcommand{\bN}{\mathbb{N}}

\newcommand\cA{\mathcal{A}}
\newcommand\cB{\mathcal{B}}
\newcommand\cC{\mathcal{C}}
\newcommand\cF{F}

\newcommand\cM{\mathcal{M}}
\newcommand\cO{\mathcal{O}}
\newcommand\cP{\mathcal{P}}
\newcommand\cR{\mathcal{R}}

\newcommand\cX{\mathcal{X}}
\newcommand\cY{\mathcal{Y}}
\newcommand\cZ{\mathcal{Z}}
\newcommand\potentialleaves{PL}

\usepackage{xspace}

\newcommand*{\defeq}{\mathrel{\vcenter{\baselineskip0.5ex \lineskiplimit0pt
                     \hbox{\scriptsize.}\hbox{\scriptsize.}}}%
                     =}

\newcommand{\parproblemdef}[4]
{
\begin{quote}
\textsc{#1}\\
\textbf{Input:} #2\\
\textbf{Parameter:} #3\\
\textbf{Question:} #4
\end{quote}
}

\theoremstyle{definition}
\newtheorem{definition}[theorem]{Definition}

\theoremstyle{remark}

\newtheorem{claim}[theorem]{Claim}
\newtheorem{observation}[theorem]{Observation}

\newtheorem*{remark*}{Remark}
\newtheorem*{question*}{Open Problem}

\newenvironment{clproof}{\begin{proof}\renewcommand{\qedsymbol}{\claimqed}}{\end{proof}\renewcommand{\qedsymbol}{\plainqed}}

\let\plainqed\qedsymbol

\begin{document}

\title{A unified polynomial-time algorithm for Feedback Vertex Set on graphs of bounded mim-width\footnote{The work was partially done while the authors were at Polytechnic University of Valencia, Spain.}}

\author[1]{Lars Jaffke\thanks{Supported by the Bergen Research Foundation (BFS).}}

\author[2]{O-joung Kwon\thanks{Supported by the European Research Council (ERC) under the European Union's Horizon 2020 research and innovation programme (ERC consolidator grant DISTRUCT, agreement No. 648527).}}

\author[1]{Jan Arne Telle}

\affil[1]{Department of Informatics, University of Bergen, Norway. \protect\\ \texttt{\{lars.jaffke, jan.arne.telle\}@uib.no}}
\affil[2]{Logic and Semantics, Technische Universit\"at Berlin, Berlin, Germany. \protect\\ \texttt{ojoungkwon@gmail.com}}

\date\today

\maketitle

\begin{abstract}
                        We give a first polynomial-time algorithm for (\textsc{Weighted}) 
\fvs\  on graphs of bounded {\em maximum induced matching width} 
(mim-width).
                Explicitly, given a branch decomposition of mim-width $w$, we give an 
$n^{\cO(w)}$-time algorithm that solves \textsc{Feedback Vertex Set}.
                This provides a unified algorithm for many well-known classes, such as 
{\sc Interval} graphs and {\sc Permutation} graphs, and furthermore,
                it gives the first polynomial-time algorithms for other classes of 
bounded mim-width, such as {\sc Circular Permutation} and \textsc{Circular $k$-Trapezoid} graphs for fixed $k$. In all these classes the 
decomposition is computable in polynomial time, as shown by Belmonte and 
Vatshelle [Theor.\ Comput.\ Sci.\ 2013].

        We show that powers of graphs of tree-width $w - 1$ or path-width $w$
        and powers of graphs of clique-width $w$ have mim-width at most $w$.
        These results extensively provide new classes of bounded mim-width. 
        We prove a slight strengthening of the first statement which implies that, surprisingly, {\sc Leaf Power} graphs which are of 
importance in the field of phylogenetic studies have mim-width at most 
$1$.
Given a tree decomposition of width $w-1$, a path decomposition of width $w$, or a clique-width $w$-expression of a graph,
        one can for any value of $k$ find a mim-width decomposition of its 
$k$-power in polynomial time, and apply our algorithm to solve 
\textsc{Feedback Vertex Set} on the $k$-power in time $n^{\cO(w)}$.

                In contrast to \fvs, we show that \textsc{Hamiltonian Cycle} is 
$\NP$-complete even on graphs of linear mim-width $1$, which further 
hints at the expressive power of the mim-width parameter.

\end{abstract}

\section{Introduction}
A feedback vertex set in a graph is a subset of its vertices whose 
removal results in an acyclic graph. The problem of finding a smallest 
such set is one of Karp's 21 famous $\NP$-complete problems \cite{Kar72} 
and many algorithmic techniques have been developed to attack this 
problem, see e.g.\ the survey~\cite{FPR99}. The study of {\sc Feedback 
Vertex Set} through the lens of parameterized algorithmics dates back to 
the earliest days of the field~\cite{DF95}
and throughout the years numerous efforts have been made to obtain 
faster algorithms for this problem~\cite{Bod94,CFL08,CNP11,DFL05,DF95,DF99,GGH06,KPS04,RSS02,RSS06}.
In terms of parameterizations by structural properties of the graph, 
{\sc Feedback Vertex Set} is e.g.\ known to be $\FPT$ parameterized by 
tree-width \cite{CNP11} and clique-width~\cite{BST13}, and $\W[1]$-hard but in 
$\XP$ parameterized by Independent Set and the size of a maximum induced 
matching~\cite{JVV14}.

In this paper, we study {\sc Feedback Vertex Set} parameterized by the 
{\em maximum induced matching width} (mim-width for short), a graph 
parameter defined in 2012 by Vatshelle~\cite{VatshelleThesis} which 
measures how easy it is to decompose a graph along vertex cuts with 
bounded maximum induced matching size on the bipartite graph induced by 
edges crossing the cut. One interesting aspect of this width-measure is 
that its modeling power is much stronger than tree-width and 
clique-width, and many well-known and deeply studied graph classes such 
as \textsc{Interval} graphs and \textsc{Permutation} graphs have 
(linear) mim-width $1$, with decompositions that can be found in 
polynomial time~\cite{BelmonteV2013,VatshelleThesis}, while their 
clique-width can be proportional to the square root of the number of 
vertices \cite{GR00}. Hence, designing an algorithm for a problem $\Pi$ that runs in 
$\XP$ time parameterized by mim-width yields polynomial-time algorithms 
for $\Pi$ on several interesting graph classes at
once.

We give an $\XP$-time algorithm for {\sc Feedback Vertex 
Set} parameterized by mim-width, assuming that a branch decomposition of 
bounded mim-width is given.\footnote{This problem was mentioned as an 
`interesting topic for further research' in \cite{JVV14}. Furthermore, the authors recently proved it to be $\W[1]$-hard~\cite{JaffkeKT2017c}.} 
Since such a decomposition can be computed in polynomial time~\cite{BelmonteV2013,VatshelleThesis} for the following classes, this 
provides a unified polynomial-time algorithm for {\sc Feedback Vertex 
Set} on all of them: \textsc{Interval} and \textsc{Bi-Interval} graphs, 
\textsc{Circular Arc}, \textsc{Permutation} and \textsc{Circular 
Permutation} graphs, \textsc{Convex} graphs, \textsc{$k$-Trapezoid}, 
\textsc{Circular $k$-Trapezoid}, \textsc{$k$-Polygon}, 
\textsc{Dilworth-$k$} and \textsc{Co-$k$-Degenerate} graphs for fixed 
$k$.
Furthermore, our algorithm can be applied to {\sc Weighted Feedback 
Vertex Set} as well, which for several of these classes was not known to 
be solvable in polynomial time. 

\begin{theorem}\label{thm:FVSmim}
	Given an $n$-vertex graph and a branch decomposition of mim-width $w$, we can solve \textsc{(Weighted) Feedback Vertex Set} in time $n^{\cO(w)}$. 	
	\end{theorem}
	
	We note that some of the above mentioned graph classes of bounded mim-width also have bounded asteroidal number, and a polynomial-time algorithm for {\sc Feedback Vertex Set} on graphs of bounded asteroidal number was previously known due to Kratsch et al.~\cite{KMT08}. However, our algorithm improves this result. For instance, {\sc $k$-Polygon} graphs have mim-width at most $2k$~\cite{BelmonteV2013} and asteroidal number $k$~\cite{SV17}. The algorithm of Kratsch et al.~\cite{KMT08} implies that {\sc Feedback Vertex Set} on {\sc $k$-Polygon} graphs can be solved in time $n^{\cO(k^2)}$ while the our result improves this bound to $n^{\cO(k)}$ time. It is not difficult to see that in general, mim-width and asteroidal number are incomparable.

	We give results that expand our knowledge of the expressive power 
of mim-width. 
	The {\em $k$-power} of a graph $G$ is the graph obtained by adding an edge $vw$ for each pair of vertices $v, w$ whose distance in $G$ is at most $k$.
We show that powers of graphs of tree-width $w-1$ or path-width $w$ and powers of graphs of clique-width $w$ have mim-width at most $w$.
        \begin{theorem}\label{thm:power}
        Given a nice tree decomposition of width $w$, all of whose join bags have size at most $w$,
        or a clique-width $w$-expression of a graph,
        one can output a branch decomposition of mim-width $w$ of its $k$-power in polynomial time.
	\end{theorem}

	Theorem~\ref{thm:power} implies 
that {\sc leaf power} graphs, of importance in the field of phylogenetic 
studies, have mim-width $1$.
 These graphs are known to be
        {\sc Strongly Chordal} and there has recently been interest in 
delineating the difference between these two graph classes, on the 
assumption that this difference was not very big \cite{Lafond17,NR16}. 
Our result actually implies a large difference, as
        it was recently shown by Mengel that there are {\sc Strongly Chordal Split} graphs of  
mim-width linear in the number of vertices \cite{Mengel16}.

        We contrast our positive result with a proof that {\sc Hamiltonian 
Cycle} is $\NP$-complete on graphs of linear mim-width $1$, even when 
given a decomposition.
        Panda and Pradhan~\cite{PP08} showed that {\sc Hamiltonian Cycle} is $\NP$-complete 
on {\sc Rooted Directed Path} graphs and we show that the 
graphs constructed in their reduction have linear mim-width $1$.
        This provides evidence that the class of graphs of linear 
mim-width $1$ is larger than one might have previously expected. Up 
until now, on all graph classes of linear mim-width $1$, {\sc 
Hamiltonian Cycle} was known to be polynomial time ({\sc Permutation}), 
or even linear time ({\sc Interval}) solvable.
  This can be compared with the fact that parameterized by clique-width, \fvs\ is $\FPT$~\cite{BST13} and \textsc{Hamiltonian Cycle} only admits an $\XP$ algorithm~\cite{BKK17,EGW01} but is $\W[1]$-hard~\cite{FGL10} (see also~\cite{FGLS14}).

Let us explain some of the essential ingredients of our dynamic 
programming algorithm.
A crucial observation is that if a forest contains no induced matching 
of size $w+1$, then the number of internal vertices of the forest is 
bounded by $6w$ (Lemma~\ref{lem:reducedforest}).
        Motivated by this observation, given a forest, we define the forest 
obtained by removing its isolated vertices and leaves to be its 
\emph{reduced forest}.
        The observation implies that in a cut $(A,B)$ of a graph $G$, 
        there are at most $\cO(n^{6w})$ possible reduced forests of some induced forests consisting of edges crossing this cut. 
        We enumerate all of them, and use these as indices of the table of our 
algorithm.
                
        However, the interaction of an induced forest $F$ in $G$ with the edges of the bipartite graph crossing the cut $(A, B)$, denote this graph by $G_{A, B}$, is not completely described by its reduced forest $R$. Observe that there might still be edges in the graph $G_{A, B}$ after removing the vertices of $R$; however, these edges are not contained in the forest $F$. We capture this property of $F$ by considering a minimal vertex cover of $G_{A, B} - V(R)$ that avoids all vertices in $F$. Hence, as a second component of the table indices, we enumerate all minimal vertex covers of $G_{A, B} - V(R)$, for any possible reduced forest $R$.
                
        To argue that the number of table entries stays bounded by 
$n^{\cO(w)}$, we use the known result that every $n$-vertex bipartite 
graph with maximum induced matching size $w$ has at most $n^{w}$ minimal 
vertex covers.
	Remark that in the companion paper~\cite{JaffkeKT2017}, we use minimal vertex covers of a bipartite graph in a similar way. 
	However, in the algorithms described in~\cite{JaffkeKT2017}, the full intersection of a solution with a cut could be used as a part of the table indices, whereas in the present paper, we can only store reduced forests (as opposed to the full forests), resulting in a more technical exposition.

        The rest of the paper is organized as follows: After giving some 
preliminary definitions and tools in Section \ref{sec:preliminaries}, in 
Section~\ref{sec:reducedforest}, we give necessary lemmas regarding 
reduced forests.
        We obtain our algorithm in Section~\ref{sec:fvs}.
        In Section~\ref{sec:hamcyc}, we prove the hardness result for \textsc{Hamiltonian Cycle}
         and in Section~\ref{sec:graphclass}, we discuss new graph classes of bounded 
mim-width.
	
	\section{Preliminaries}\label{sec:preliminaries}

	For integers $a$ and $b$ with $a \le b$, we let $[a..b] \defeq \{a, a+1, \ldots, b\}$ and if $a$ is positive, we define $[a] \defeq [1..a]$. 
	Every graph in this paper is finite, undirected and simple.
	For a graph $G$ we denote by $V(G)$ and $E(G) \subseteq {V(G) \choose 2}$ its vertex and edge set, respectively.
	For graphs $G$ and $H$ we say that $G$ is a {\em subgraph} of $H$, if $V(G) \subseteq V(H)$ and $E(G) \subseteq E(H)$. 
	For a vertex set $X \subseteq V(G)$, we denote by $G[X]$ the subgraph {\em induced} by $X$, i.e.\ $G[X] \defeq (X, E(G) \cap {X \choose 2})$. 
	We use the shorthand $G - X$ for $G[V(G) \setminus X]$. 
		For two graphs $G_1$ and $G_2$,  $G_1\cup G_2$ is the graph with the vertex set $V(G_1) \cup V(G_2)$ and the edge set $E(G_1) \cup E(G_2)$, 
and $G_1\cap G_2$ is the graph with the vertex set $V(G_1) \cap V(G_2)$ and the edge set $E(G_1) \cap E(G_2)$.
	For a vertex $v \in V(G)$, we denote by $N_G(v)$ the set of {\em neighbors} of $v$ in $G$, i.e.~$N_G(v) \defeq \{w \in V(G) \mid \{v, w\} \in E(G)\}$, 
	and the number of neighbors of $v$ is called its \emph{degree}, denoted by $\deg_G(v) \defeq \card{N_G(v)}$.
	For $A\subseteq V(G)$, let $N_G(A)$ be the set of vertices in $V(G)\setminus A$ having a neighbor in $A$.
	We drop $G$ as a subscript if it is clear from the context.
	We denote by $\mathcal{C}(G)$ the set of connected components of $G$.

For two (disjoint) vertex sets $X, Y \subseteq V(G)$, 
	we denote by $G[X, Y]$ the bipartite subgraph of $G$ with bipartition $(X, Y)$ such that for $x\in X, y\in Y$, $x$ and $y$ are adjacent in $G[X, Y]$ if and only if they are adjacent in $G$. A {\em cut} of $G$ is a bipartition $(A, B)$ of its vertex set.
	A set $M$ of edges is a \emph{matching} if no two edges in $M$ share an endpoint, and a matching $\{a_1b_1, \ldots, a_kb_k\}$ is  \emph{induced} if there are no other edges in the subgraph induced by 
	$\{a_1, b_1, \ldots, a_k, b_k\}$. 
	A vertex set $S \subseteq V(G)$ is a \emph{vertex cover} of $G$ if every edge in $G$ is incident with a vertex in $S$.

	For $r \in \bN$, a graph $G$ is called {\em $r$-regular} if $\deg_G(v) = r$ for all $v \in V(G)$. A connected $2$-regular graph is called a {\em cycle}. A graph that does not contain a cycle as a subgraph is called a {\em forest} and a connected forest is a {\em tree}. A tree of maximum degree is called a {\em path} and we refer to the {\em length} of a path as the number of its edges.

	A {\em star} is a tree on at least three vertices containing a special vertex, called its {\em central} vertex, adjacent to all other vertices.
	We require a star to have at least three vertices to emphasize the distinction between a star and a graph consisting of a single edge, as they require different treatment in our algorithm.

\subsection{Parameterized Complexity}

We now give the basic definition in parameterized complexity and refer to \cite{Cygan15,DowneyF13} for an introduction.

\begin{definition}[Parameterized Problem, $\FPT$, $\XP$]
	Let $\Sigma$ be an alphabet. A {\em parameterized problem} is a set $\Pi \subseteq \Sigma^* \times \mathbb{N}$, the second component being the {\em parameter} which usually expresses a structural measure of the input.
	\begin{enumerate}
		\item A parameterized problem $\Pi$ is {\em fixed-parameter tractable} ($\FPT$) if there exists an algorithm that for any $\langle x, k \rangle \in \Sigma^* \times \mathbb{N}$ decides whether $\langle x, k \rangle \in \Pi$ in time $f(k) \cdot |x|^{\cO(1)}$, for some computable function $f$.
		\item A parameterized problem $\Pi$ is in $\XP$ if there exists an algorithm that for any $\langle x, k \rangle \in \Sigma^* \times \mathbb{N}$ decides whether $\langle x, k \rangle \in \Pi$ in time $f(k) \cdot |x|^{g(k)}$, for some computable functions $f$ and $g$.
	\end{enumerate}
\end{definition}

\newcommand*{\dectree}{T}
\newcommand*{\decf}{\mathcal{L}}
\newcommand{\crossinggraph}[1]{G_{#1, \bar{#1}}}
\newcommand{\crossinggraphAB}[2]{G_{#1, #2}}	

\subsection{Branch Decompositions and Mim-Width}

	For a graph $G$ and a vertex subset $A$ of $G$, we define $\mimval_G(A)$ to be the maximum size of an induced matching in $G[A, V(G) \setminus A]$. 

A pair $(\dectree, \decf)$ of a subcubic tree $\dectree$ and a bijection $\decf$ from $V(G)$ to the set of leaves of $\dectree$ is called a \emph{branch decomposition}.
For each edge $e$ of $\dectree$, 
let $\dectree^e_1$ and $\dectree^e_2$ be the two connected components of $\dectree-e$, and 
let $(A^e_1, A^e_2)$ be the vertex bipartition of $G$ such that for each $i\in \{1,2\}$, 
$A^e_i$ is the set of all vertices in $G$ mapped to leaves contained in $\dectree^e_i$ by $\decf$. 
The {\em mim-width of $(\dectree, \decf)$}, denoted by $\mimw(\dectree, \decf)$, is defined as $\max_{e \in E(\dectree)} \mimval_G(A^e_1)$.
The minimum mim-width over all branch decompositions of $G$ is called the {\em mim-width of $G$}, and the {\em linear mim-width of $G$} if $\dectree$ is restricted to a path with a pendant leaf at each node.
If $\abs{V(G)}\le 1$, then $G$ does not admit a branch decomposition, and the mim-width of $G$ is defined to be $0$.

To avoid confusion, we refer to elements in $V(T)$ as {\em nodes} and elements in $V(G)$ as {\em vertices} throughout the rest of the paper.
Given a branch decomposition, one can subdivide an arbitrary edge and let the newly created vertex be the root of $\dectree$, in the following denoted by $r$. Throughout the following we assume that each branch decomposition has a root node of degree two. 
For two nodes $t, t' \in V(T)$, we say that $t'$ is a {\em descendant} of $t$ if $t$ lies on the path from $r$ to $t'$ in $T$.
For $t \in V(\dectree)$, we denote by $G_t$ the subgraph induced by all vertices that are mapped to a leaf that is a descendant of $t$, i.e.\ $G_t = G[X_t]$, where $X_t = \{v \in V(G) \mid \decf^{-1}(t') = v \mbox{ where } t' \mbox{ is a descendant of $t$ in $\dectree$}\}$. We use the shorthand $V_t$ for $V(G_t)$ and let $\bar{V_t} \defeq V(G) \setminus V_t$.

The following definitions which we relate to branch decompositions of graphs will play a central role in the design of the algorithms in Section~\ref{sec:fvs}.

\begin{definition}[Boundary]
	Let $G$ be a graph and $A, B \subseteq V(G)$ such that $A \cap B = \emptyset$. We let $\bd_B(A)$ be the set of vertices in $A$ that have a neighbor in $B$, i.e.\ $\bd_B(A) \defeq \{v \in V(A) \mid N(v) \cap B \neq \emptyset\}$. We define $\bd(A) \defeq \bd_{V(G) \setminus A}(A)$ and call $\bd(A)$ the {\em boundary} of $A$ in $G$.
\end{definition}

\begin{definition}[Crossing Graph]
	Let $G$ be a graph and $A, B \subseteq V(G)$. If $A \cap B = \emptyset$, we define the graph $G_{A, B} \defeq G[\bd_B(A), \bd_A(B)]$ to be the {\em crossing graph from $A$ to $B$}.
\end{definition}
	
	If $(\dectree, \decf)$ is a branch decomposition of $G$ and $t_1, t_2 \in V(\dectree)$ such that $V_{t_1}\cap V_{t_2}=\emptyset$, we use the shorthand $\crossinggraphAB{t_1}{t_2} \defeq \crossinggraphAB{V_{t_1}}{V_{t_2}}$. We use the analogous shorthand notations $\crossinggraphAB{t_1}{\bar{t_2}} \defeq \crossinggraphAB{V_{t_1}}{\bar{V_{t_2}}}$ and $\crossinggraphAB{\bar{t_1}}{t_2} \defeq \crossinggraphAB{\bar{V_{t_1}}}{V_{t_2}}$ (whenever these graphs are defined). For the frequently arising case when we consider $\crossinggraph{t}$ for some $t \in V(\dectree)$, we refer to this graph as the {\em crossing graph w.r.t.\ $t$}.

\newcommand{\nummis}{\mathrm{mis}}
\newcommand{\nummvc}{\mathrm{mvc}}

\subsection{The Minimal Vertex Covers Lemma} 
	Let $G$ be a graph. 
	We prove that given a set $A \subseteq V(G)$, the number of minimal vertex covers in $G_{A, V(G) \setminus A}$ is bounded by $n^{\mimval_G(A)}$, and 
	furthermore, the set of all minimal vertex covers can be enumerated in time $n^{\mathcal{O}(\mimval_G(A))}$. 
	This observation is crucial to argue that in our dynamic programming algorithm, there are at most $n^{\cO(w)}$ table entries to consider at each node of the given branch decomposition $(T, \decf)$, where $w$ denotes the mim-width of $(T, \decf)$.
	Notice that the bound on the number can be easily obtained by combining two results, \cite[Lemma 1]{BelmonteV2013} and \cite[Theorem 3.5.5]{VatshelleThesis};
	however, an enumeration algorithm is not given explicitly. To be self-contained, we state and prove it here.

    \begin{corollary}[Minimal Vertex Covers Lemma]\label{cor:mvc:lemma}
		Let $H$ be a bipartite graph on $n$ vertices with a bipartition $(A,B)$. 
		The number of minimal vertex covers of $H$ is at most $n^{\mimval_H(A)}$, 
		and the set of all minimal vertex covers of $H$ can be enumerated in time $n^{\mathcal{O}(\mimval_H(A))}$.
    \end{corollary}
    \begin{proof}
    Let $w\defeq \mimval_H(A)$.
    For each vertex set $R\subseteq A$ with $\abs{R}\le w$, 
    let $X_R\subseteq A$ be the set of all vertices having a neighbor in $B\setminus N(R)$.
    We enumerate the sets in 
    \[\mathcal{M}=\{N(R)\cup X_R   :R\subseteq A, \abs{R}\le w\}.\]
    Clearly, we can enumerate them in time $n^{\mathcal{O}(w)}$.
    It is not difficult to see that each set in $\mathcal{M}$ is a minimal vertex cover.
    We claim that $\mathcal{M}$ is the set of all minimal vertex covers in $H$.
    
    We use the result by Belmonte and Vatshelle~\cite{BelmonteV2013} that
    for a graph $G$ and $A\subseteq V(G)$, $\mim(A)\le k$ if and only if for every $S\subseteq A$, there exists $R\subseteq S$ such that
  $N(R)\cap (V(G)\setminus A)=N(S)\cap (V(G)\setminus A)$ and $\abs{R}\le k$.
    
    Let $U$ be a minimal vertex cover of $H$.
    Clearly, every vertex in $A\setminus U$ has no neighbors in $B\setminus U$, as $U$ is a vertex cover.
    Therefore, by the result of Belmonte and Varshelle, 
    there exists $R\subseteq A\setminus U$ such that $\abs{R}\le w$ and $N(R)\cap B=N(A\setminus U)\cap B=U\cap B$.
    Clearly, $U\cap A=X_R$; if a vertex in $U\cap A$ has no neighbors in $B\setminus U$, then we can remove it from the vertex cover.
    Therefore, $U\in \mathcal{M}$, as required.    \end{proof}

\newcommand{\crosg}[1]{\crossinggraph{#1}}
\newcommand{\reduced}{\mathfrak{R}}
\newcommand{\leaves}{L}
\newcommand{\noncrossing}{\mathrm{NC}}

\section{Reduced forests}\label{sec:reducedforest}

We formally introduce the notion of a {\em reduced forest} which will be crucial to obtain the desired runtime bound of the algorithm for {\sc Feedback Vertex Set}. 

\begin{definition}[Reduced Forest]\label{def:reduced:forest}
	Let $F$ be a forest. A {\em reduced forest} of $F$, denoted by $\reduced(F)$, is an induced subforest of $F$ obtained as follows.
	\begin{enumerate}[(i)]
		\item Remove all isolated vertices of $F$.
		\item For each component $C$ of $F$ with $\card{V(C)} = 2$, remove one of its vertices. 
	\item For each component $C$ of $F$ with $\card{V(C)} \ge 3$, remove all leaves of $C$.
		\end{enumerate}
\end{definition}

	Note that if $F$ has no component that is an edge (i.e. $\card{V(C)} = 2$) then the reduced forest is uniquely defined.
	We give an upper bound on the size of a reduced forest $\reduced(F)$ by a function of the size of a maximum induced matching in the forest $F$. 
	\begin{lemma}\label{lem:reducedforest}
	Let $p$ be a positive integer. 
	If $F$ is a forest whose maximum induced matching has size at most $p$,  
	then $\abs{\reduced(F)} \le 6p$.
	\end{lemma}
	\begin{proof}
	For a forest $F$, we denote by $m(F)$ the size of the maximum induced matching in $F$.
	We prove the lemma by induction on $m(F)$.
	We may assume $F$ contains no isolated vertices, as they will be removed in the reduced forest.
	If $m(F)=0$, then $F$ contains no edges, and we are done.
	If $m(F)=1$, then $F$ consists of one component that contains no path of length $4$ 
	which implies that $\reduced(F)$ contains at most $2$ vertices.
	We may assume $m(F)=p>1$.
	
	Suppose $F$ contains a connected component $C$ containing no path of length $4$.
	As observed, $C$ contains no induced matching of size larger than one.
	Since $C$ contains an edge, we have $m(F-V(C))=m(F)-1$.
	By the induction hypothesis, $\reduced(F-V(C))$ contains at most $6(p-1)$ vertices, and 
	we have that $\reduced(F)$ contains at most $6(p-1)+2\le 6p$ vertices.
	We may assume every component $C$ of $F$ contains a path of length $4$, implying that 
	$\reduced(C)$ contains at least $3$ vertices.
	
	Now, suppose $F$ contains a path $v_1v_2v_3v_4v_5$ such that 
	$v_1$ and $v_5$ are not leaves of $F$, and $v_2, v_3, v_4$ have degree $2$ in $\reduced(F)$.
	Let $F'$ be the forest obtained from $F$ by removing $v_2, v_3, v_4$ and adding an edge $v_1v_5$.
	We observe that $m(F')\le m(F)-1$.
	Let $M$ be a maximum induced matching of $F'$. If $M$ contains the edge $v_1v_5$, then 
	we can obtain an induced matching for $F$ by removing $v_1v_5$ and adding $v_1v_2$ and $v_4v_5$.
	If $M$ does not contain $v_1v_5$, then one of $v_1$ and $v_5$ is not matched by $M$.
	Then we can select one of $v_2v_3$ and $v_3v_4$ to increase the size of an induced matching.
	Thus, we have $m(F')\le m(F)-1$.
	By the induction hypothesis, $\reduced(F')$ contains at most $6(p-1)$ vertices, 
	and thus $\reduced(F)$ contains at most $6(p-1)+3=6p-3$ vertices.
	We may assume there is no such path.
	
	Let $C$ be a component of $F$. As $\reduced(C)$ contains at least $3$ vertices, 
	the leaves of $\reduced(C)$ form an independent set. 
	Let $t$ be the number of leaves in $\reduced(C)$.
	Since each leaf of $\reduced(C)$ is incident with a leaf of $C$, 
	$C$ contains an induced matching of size at least $t$.
	Thus, $m(F-V(C))\le m(F)-t$.
	Note that $\reduced(C)$ contains at most $t$ vertices of degree at least $3$. 
	Also, by the previous argument, 
	there are at most $2$ vertices between two vertices of degree other than $2$ in $\reduced (C)$.
	Thus, $\reduced(C)$ contains at most $t+t+2(2t-1)\le 6t$ vertices. 
	Therefore, the result follows by the induction hypothesis.
	\end{proof}

	Let $(A,B)$ be a vertex partition of a graph $G$.
	Let $R$ be some forest in $G_{A,B}$. 
	In the algorithm, we will be asking if there exists an induced forest $F$ in $G[A\cup \bd(B)]$ such that 
	$F\cap G_{A,B}$ has $R$ as a reduced forest.
	However, this formulation turns out to be quite technical, as we need to significantly consider some edges in $B$ when we merge two partial solutions.
	To ease this task in the dynamic programming algorithm, we define the following notion on an induced forest in $G[A\cup \bd(B)]-E(G[\bd(B)])$.

\begin{definition}[Forest respecting a forest and a minimal vertex cover]\label{def:respect}
	Let $(A, B)$ be a vertex partition of a graph $G$.
	Let $R$ be an induced forest in $G_{A,B}$ and $M$ be a minimal vertex cover 
	of $G_{A, B}-V(R)$. 
	An induced forest $F$ in $G[A\cup \bd(B)]-E(G[\bd(B)])$ \emph{respects} $(R, M)$ if it satisfies the following.
	\begin{enumerate}[(i)]
	\item $R$ is a reduced forest of $G_{A,B}\cap F$.
	\item $V(F) \cap M = \emptyset$.
\end{enumerate}
\end{definition}

Suppose $R$ is an induced forest in $G_{A,B}$. For an induced forest $F$ of $G$ containing $V(R)$, 
	there are two necessary conditions for $R$ to be a reduced forest of $F\cap G_{A,B}$. 
	First, if $F\cap G_{A,B}$ contains a vertex $v$ in $G_{A,B}-V(R)$ having at least two neighbors in $R$, 
	then $v$ should be contained in the reduced forest.
	Therefore, in $F\cap G_{A,B}$, every vertex in $V(F\cap G_{A,B})\setminus V(R)$ should have at most one neighbor in $R$.
	Second, every leaf $x$ of $R$ should have a neighbor $y$ in $G_{A,B}-V(R)$ such that the only neighbor of $y$ in $R$ is $x$; 
	otherwise, we would have removed $x$ when taking a reduced forest. 
	
	Motivated by this observation we define the notion of potential leaves, which is a possible leaf neighbor of some vertex in $V(R)$.

\begin{definition}[Potential Leaves]
		Let $(A, B)$ be a vertex partition of a graph $G$.
	Let $R$ be an induced forest in $H\defeq G_{A,B}$ and $M$ be a minimal vertex cover 
	of $H-V(R)$.
	For each vertex $x \in V(R)$, we define its set of {\em potential leaves}, denoted by $\potentialleaves_{R,M}(x)$, as 
	\begin{align*}
		\potentialleaves_{R,M}(x) \defeq N_H(x) \setminus N_H(V(R)\setminus \{x\}) \setminus (M\cup V(R)).
	\end{align*}
\end{definition}

	We can observe the following.
	\begin{observation}
	Every forest $F$ respecting $(R, M)$ 
	should contain at least one vertex in $\potentialleaves_{R,M}(x)$ for each leaf $x$ of $R$.
	\end{observation}

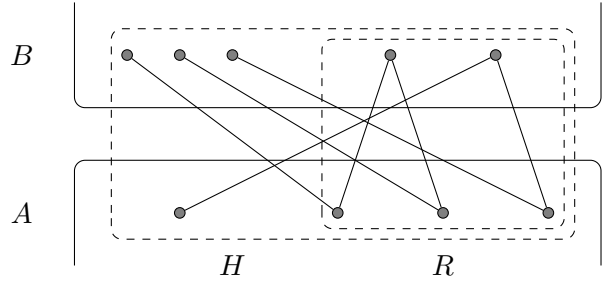
\begin{figure}
  \centering
  \begin{tikzpicture}[scale=0.7]
  \tikzstyle{w}=[circle,draw,fill=black!50,inner sep=0pt,minimum width=4pt]

	\draw[rounded corners] (0, -2)--(0,0)--(10,0)--(10,-2);
	\draw[rounded corners] (0, 3)--(0,1)--(10,1)--(10,3);
     \node at (-1, -1) {$A$};
     \node at (-1, 2) {$B$};
    \node at (7, -2) {$R$};

      \draw (2,-1) node [w] (v1) {};
 \draw (5,-1) node [w] (v2) {};
 \draw (7,-1) node [w] (v3) {};
 \draw (9,-1) node [w] (v4) {};
 \draw (1,2) node [w] (w0) {};
 \draw (2,2) node [w] (w1) {};
 \draw (3,2) node [w] (w2) {};
 \draw (6,2) node [w] (w3) {};
\draw (8,2) node [w] (w4) {};

	\draw (v2)--(w3)--(v3); \draw(v4)--(w4);
	\draw (w0)--(v2);
	\draw (w1)--(v3);
	\draw (w2)--(v4);
		\draw (v1)--(w4);
		\draw[dashed, rounded corners] (4.7, -1)--(4.7,2.3)--(9.3,2.3)--(9.3,-1.3)--(4.7,-1.3)--(4.7, -1);
\draw[dashed, rounded corners] (0.7, -1)--(0.7,2.5)--(9.5,2.5)--(9.5,-1.5)--(0.7,-1.5)--(0.7, -1);
    \node at (3, -2) {$H$};

   \end{tikzpicture}     \caption{The graph $R$ is a reduced forest of $H$. }\label{fig:setting}
\end{figure}

	For a subset $A'$ of $A$, we consider a pair of an induced forest $R'$ and a minimal vertex cover $M'$ of $G_{A', V(G)\setminus A'}-V(R')$
	and we say that this pair is a restriction of a pair of $R$ and $M$ for $A$, if they satisfy certain natural properties.  
	In the dynamic programming algorithm, we use this notion to study the structure a partial solution w.r.t.\ a cut corresponding to a node $t$ in the branch decomposition induces on the cuts corresponding to the children of $t$.

\begin{definition}[Restriction of a reduced forest and a minimal vertex cover]
	Let $(A_1, A_2, B)$ be a vertex partition of a graph $G$.
	Let $R$ be an induced forest in $G_{A_1\cup A_2,B}$ and $M$ be a minimal vertex cover 
	of $G_{A_1\cup A_2, B}-V(R)$. 
	An induced forest $R_1$ in $G_{A_1, A_2\cup B}$ and a minimal vertex cover $M_1$ of $G_{A_1, A_2\cup B}-V(R_1)$ are \emph{restrictions} of $R$ and $M$ to $G_{A_1, A_2\cup B}$
	if they satisfy the following:
	\begin{enumerate}
	\item $V(R)\cap A_1\subseteq V(R_1)$ and for every $v\in V(R)\cap B$ having at least two neighbors in $V(R)\cap A_1$, 
$v\in V(R_1)$.
	\item $(V(R_1)\setminus V(R))\cap B=\emptyset$ and $V(R_1)\cap M=\emptyset$.
	\item Every vertex in $(V(R_1)\setminus V(R))\cap A_1$ has at most one neighbor in $V(R)\cap B$.
	\item $V(R)\cap M_1=\emptyset$ and $M\cap A_1\subseteq M_1$.
	\item Let $v$ be a vertex in $M\cap B$ incident with an edge $vw$ in $G_{A_1, B}-V(R)$ for some $w\notin V(R_1)$ that is not covered by any vertices in $M\setminus \{v\}$.
	Then either $v\in M_1$ or $w\in M_1$.
	\end{enumerate}
\end{definition}

	Lastly, we define a notion for merging two partial solutions.

\begin{definition}[Compatibility]
	Let $(A_1, A_2, B)$ be a vertex partition of a graph $G$.
	Let $R$ be an induced forest in $G_{A_1\cup A_2, B}$, and 
	for each $i\in \{1, 2\}$, let $R_i$ be an induced forest in $G_{A_i, A_{3-i}\cup B}$, and  
	$P_i$ be a partition of $\mathcal{C}(R_i)$.
	We construct an auxiliary graph $Q$ with respect to $(R, R_1, R_2, P_1, P_2)$ in $G$ as follows.
	Let $Q$ be the graph on the vertex set $\mathcal{C}(R)\cup \mathcal{C}(R_1)\cup \mathcal{C}(R_2)$ such that
	\begin{itemize}
	\item for $H_1$ and $H_2$ contained in distinct sets of $\mathcal{C}(R), \mathcal{C}(R_1), \mathcal{C}(R_2)$, $H_1$ is adjacent to $H_2$ in $Q$ 
	if and only if $V(H_1)\cap V(H_2)\neq \emptyset$, 
	\item for $H_1, H_2\in \mathcal{C}(R_i)$, $H_1$ is adjacent to $H_2$ if and only if $H_1$ and $H_2$ are contained in the same part of $P_i$, 
	\item $\mathcal{C}(R)$ is an independent set. 
		\end{itemize}
		We say that the tuple $(R, R_1, R_2, P_1, P_2)$ is \emph{compatible} in $G$ if $Q$ has no cycles.
		We define $\mathcal{U}(R, R_1, R_2, P_1, P_2)$ to be the partition of $\mathcal{C}(R)$ such that
		for $H_1, H_2\in \mathcal{C}(R)$, $H_1$ and $H_2$ are contained in the same part 
		if and only if they are contained in the same connected component of $Q$.
\end{definition}

The remainder of this section is devoted to proving several technical propositions related to the notions introduced above that will be important to establish the correctness of the algorithm proposed in Section \ref{sec:fvs}. Let $t \in V(T)$ be a no-leaf node in the given branch decomposition of $G$ with children $a$ and $b$. 
In Section \ref{sec:reducedforest:top:bottom} we show that given any forest $F_t$ in $G[V_t \cup \bd(\bar{V_t})]$ respecting $(R_t, M_t)$, where $R_t$ is an induced forest in $G_{t, \bar{t}}$ and $M_t$ a minimal vertex cover of $G_{t, \bar{t}} - V(R_t)$, we can find restrictions $(R_a, M_a)$ and $(R_b, M_b)$ to $G_{a, \bar{a}}$ and $G_{b, \bar{b}}$, respectively, such that a forest $F_a$ in $G[V_a \cup \bd(\bar{V_a})]$ respecting $(R_a, M_a)$ and a forest $F_b$ in $G[V_b \cup \bd(\bar{V_b})]$ respecting $(R_b, M_b)$ can be combined to the forest $F_t$. In other words, $(R_t, R_a, R_b, P_a, P_b)$, where $P_a$ and $P_b$ denote the partitions induced by $F_a$ and $F_b$, respectively, is compatible. In Section \ref{sec:reducedforest:bottom:top} we prove the converse direction. For the sake of generality, we will state the results in terms of a $3$-partition $(A_1, A_2, B)$ rather than $(V_a, V_b, V_t)$ (i.e.\ independently of a branch decomposition of a graph).

	\subsection{Top to bottom}\label{sec:reducedforest:top:bottom}

	\begin{proposition}\label{prop:restriction1}
	Let $(A_1, A_2, B)$ be a vertex partition of a graph $G$.
	Let $R$ be an induced forest in $G_{A_1\cup A_2, B}$  and $M$ be a minimal vertex cover 
	of $G_{A_1\cup A_2, B}-V(R)$. 
	Let $H$ be an induced forest in $G[A_1\cup A_2\cup \bd(B)]-E(G[\bd(B)])$ respecting $(R, M)$.
	
	Then there are restrictions $(R_1, M_1)$ and $(R_2, M_2)$ of $(R, M)$ to $G_{A_1, A_2\cup B}$ and $G_{A_2, A_1\cup B}$, respectively, such that 
	\begin{enumerate}[(i)]
	\item for each $i\in \{1, 2\}$, $H\cap G[A_i\cup \bd(A_{3-i}\cup B)]-E(G[\bd(A_{3-i}\cup B)])$ respects $(R_i, M_i)$, \label{prop:restriction1:1}
	\item every vertex in $(V(R)\setminus (V(R_1)\cup V(R_2))) \cap B$ has at least two neighbors in $(V(R_1)\cap A_1)\cup (V(R_2)\cap A_2)$.\label{prop:restriction1:2}
	\end{enumerate}
	\end{proposition}
	\begin{proof}
	For each $i\in \{1,2\}$, let $F_i^*\defeq H\cap G[A_i\cup \bd(A_{3-i}\cup B)]-E(G[\bd(A_{3-i}\cup B)])$, and 
	let $F_i\defeq F_i^*\cap G_{A_i, A_{3-i}\cup B}$, and let $R_i$ be a reduced forest of $F_i$ such that the following holds.
	\begin{description}
		\item[(Single-edge Rule.)] For a single-edge component $vw$ of $F_i$ with $v\in V(R)$ and $w\notin V(R)$, 
	we select $v$ as a vertex of $R_i$.
	\end{description}
	We first prove (\ref{prop:restriction1:2}).	
	\begin{claim}\label{claim:twoneighbors}
	Every vertex in $(V(R) \setminus (V(R_1) \cup V(R_2))) \cap B$ has at least two neighbors in $(V(R_1)\cap A_1)\cup (V(R_2)\cap A_2)$.
	\end{claim}
	\begin{clproof}
	Suppose there exists a vertex $v$ in $(V(R) \setminus (V(R_1) \cup V(R_2))) \cap B$ having at most one neighbor in $(V(R_1)\cap A_1)\cup (V(R_2)\cap A_2)$.
	If $v$ had only one neighbor in $V(H) \cap (A_1 \cup A_2)$, then $vw$ was a single-edge component; otherwise, $v$ would have been removed while taking the reduced forest of $H$. But then the Single-edge rule forces $v \in V(R_1) \cup V(R_2)$, a contradiction with the assumption.
	So $v$ has at least two neighbors in $V(H)\cap (A_1\cup A_2)$.
	Thus, $v$ has a neighbor not contained in $(V(R_1)\cap A_1)\cup (V(R_2)\cap A_2)$.
	Let $w$ be such a vertex, and without loss of generality, we assume $w\in A_1$.
	
	If $v$ has a neighbor other than $w$ in $V(H)\cap A_1$, then $v$ is contained in $R_1$.
	So, in $H$, $w$ is the unique neighbor of $v$ in $V(H)\cap A_1$.
	Also, since $w\notin V(R_1)$, $v$ is the unique neighbor of $v$ in $F_1$.
	Then $vw$ is a single-edge component of $F_1$, and by the Single-edge Rule, we selected $v$ as a vertex of $R_1$.
	This contradicts $v\notin V(R_1)$.

	We conclude that every vertex in $(V(R) \setminus (V(R_1) \cup V(R_2))) \cap B$ has at least two neighbors in $(V(R_1)\cap A_1)\cup (V(R_2)\cap A_2)$.
	\end{clproof}
		
	In the remainder of this proof we show (\ref{prop:restriction1:1}), i.e.\ that for each $i\in \{1, 2\}$, $R_i$ is a restriction of $R$ that respects $F_i^*$.
	We give the proof for $i=1$; an analogous proof holds for $i=2$.
	We first verify the first condition of being a restriction.

		\begin{claim}\label{claim:condition1}
		$V(R)\cap A_1\subseteq V(R_1)$ and for every $v\in V(R)\cap B$ having at least two neighbors in $V(R)\cap A_1$, 
$v\in V(R_1)$.
		\end{claim}
		\begin{clproof}
		Let $v\in V(R)\cap A_1$. Then either $v$ has degree at least $2$ in $F_1$ or 
		the unique neighbor of $v$ in $F_1$ is its potential leaf in $H$. In the former case, clearly $v$ is contained in $R_1$, 
		and in the latter case, $v$ was chosen as a vertex of $R_1$ by Single-edge Rule.
		If $v\in V(R)\cap B$ has at least two neighbors in $V(R)\cap A_1$, then clearly $v\in V(R_1)$, as all such neighbors are in $F_1$.
		\end{clproof}

	We now verify the second condition of being a restriction.
	\begin{claim}\label{claim:condition2}
		\begin{enumerate}[(a)]
			\item $(V(R_1)\setminus V(R))\cap B=\emptyset$.\label{claim:condition2:1}
			\item $V(R_1)\cap M=\emptyset$.\label{claim:condition2:2}
		\end{enumerate}
	\end{claim}
	\begin{clproof}
	 (\ref{claim:condition2:1}) It is sufficient to prove that 
	every vertex in $(V(F_1)\setminus V(R))\cap B$ is not contained in $R_1$.
	Suppose $v$ is a vertex in $(V(F_1)\setminus V(R))\cap B$.
	If $v$ has degree at least $2$ in $H$, then $v \in V(R)$, so we can assume that $v$ has degree at most $1$ in $H$.
	If $v$ is isolated in $F_1$, then $v \notin V(R_1)$, so $v$ has degree $1$ in $F_1$.	
	Let $w$ be the neighbor of $v$ in $F_1$. 
	If $w$ has degree at least $2$ in $F_1$, then $v$ is removed by definition of a reduced forest.
	If $vw$ is a single-edge component, then by the Single-edge Rule, 
	$v\notin V(R_1)$ and $w\in V(R_1)$.
	We conclude that 
	$(V(R_1)\setminus V(R))\cap B=\emptyset$.
	
	(\ref{claim:condition2:2}) As $R$ avoids $M$, clearly, $R_1$ also avoids $M$.
	\end{clproof}
	
	We also verify the third condition.
	
	\begin{claim}\label{claim:oneneighbor}
	Every vertex in $(V(R_1)\setminus V(R))\cap A_1$ has at most one neighbor in $V(R)\cap B$.
	\end{claim}
	\begin{clproof}
	Suppose not and let $v \in (V(R_1) \setminus V(R)) \cap A_1$ such that $v$ has two neighbors $x$ and $y$ in $V(R) \cap B$. Clearly, $\{v, x, y\} \subseteq V(H)$. But then, $v \in V(R)$ by the definition of reduced forests, a contradiction.
	\end{clproof}

	We now construct a minimal vertex cover $M_1$ of $G_{A_1, A_2\cup B}-V(R_1)$, and verify the fourth and fifth conditions of being a restriction.
	Let $M'$ be the set of all vertices $v$ in $M$ such that $v$ is incident with an edge $vw$ in $G_{A_1, A_2}-V(R)$ 
	where $vw$ is not covered by any vertices in $M\setminus \{v\}$ and $w\notin V(R_1)$.

\begin{figure}
  \centering
  \begin{tikzpicture}[scale=0.7]
  \tikzstyle{w}=[circle,draw,fill=black!50,inner sep=0pt,minimum width=4pt]

	\draw[rounded corners] (-5,0)--(4.5,0)--(4.5, -5);
	\draw[rounded corners] (15,0)--(5.5,0)--(5.5, -5);
	\draw[rounded corners] (-3, 3)--(-3,1)--(13,1)--(13,3);
     \node at (-5, -1) {$A_1$};
     \node at (15, -1) {$A_2$};
     \node at (-4, 2) {$B$};

	\draw[red, thick, dashed, rounded corners] (3, -1)--(3,2.3)--(7,2.3)--(7,-1.3)--(3,-1.3)--(3, -1);
    \draw[blue, thick, dashed, rounded corners] (2, -1)--(2,2.5)--(5,2.5)--(5,-2.3)--(2,-2.3)--(2, -1);
    \node at (7.3, 0.5) {$R$};
	\node at (3.5, 2.8) {$R_1$};

 	\draw[thick, rounded corners] (6, -2.5)--(6,-2)--(8,-2)--(8,-4)--(6,-4)--(6, -2.5);
 	\draw[thick, dotted](7,-2.5-.5) [in=30,out=150] to (3,-3-.5);
	\draw[thick, dotted](7,-3-.5) [in=30,out=150] to (3,-3.5-.5);

 	\draw[thick, rounded corners] (-1, -2.5)--(-1,-2)--(1,-2)--(1,-4)--(-1,-4)--(-1, -2.5);
 	\draw[thick, dotted](0,-2.5-.5) [in=-150,out=60] to (6,2);
	\draw[thick, dotted](0,-3-.5) [in=-150,out=60] to (6,1.5);

 	\draw[thick, rounded corners] (-2.5, 1)--(-2.5,2)--(-.5,2)--(-.5,-1)--(-2.5,-1)--(-2.5, 1);
 	\draw[thick, dotted](-1.5,1.7) -- (.5,-0.7);
	\draw[thick, dotted](-1.5,1.3) -- (.5,-1.1);
 	\draw[thick, dotted](-1.5,-.7) -- (.5,1.7);
	\draw[thick, dotted](-1.5,-.3) -- (.5,2.1);
	
   \node at (7, -4.3) {$Y$};
   \node at (0, -4.3) {$Z$};
   \node at (-1.5, -1.3) {$M'$};

   \end{tikzpicture}     \caption{An illustration of $M', Y, Z$ in Claim~\ref{claim:condition4}. }\label{fig:mvc}
\end{figure}
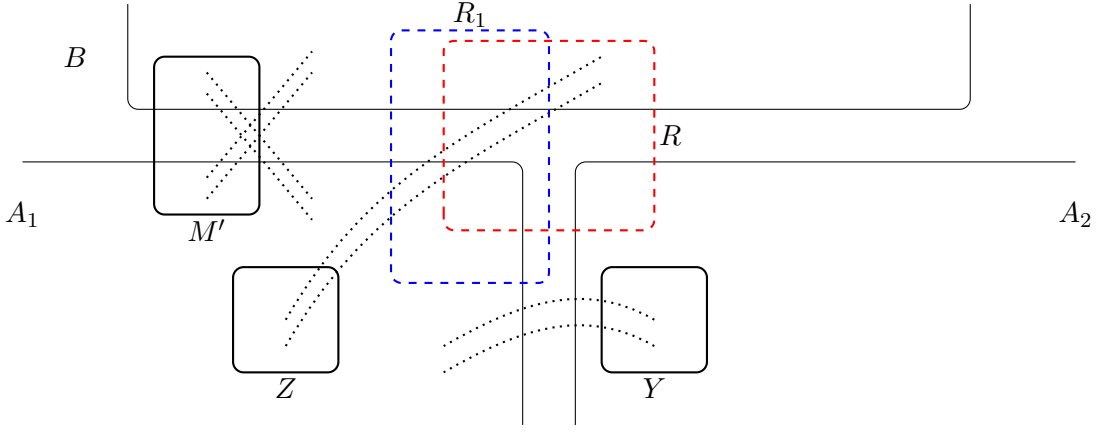

	\begin{claim}\label{claim:condition4}
	There is a minimal vertex cover $M_1$ of $G_{A_1, A_2\cup B}-V(R_1)$ satisfying the following.
	\begin{itemize}
	\item $V(R)\cap M_1=\emptyset$ and $M\cap A_1 \subseteq M_1$.
	\item Let $v$ be a vertex in $M\cap B$ incident with an edge $vw$ in $G_{A_1, B}-V(R)$ for some $w\notin V(R_1)$ that is not covered by any vertices in $M\setminus \{v\}$.
	Then either $v\in M_1$ or $w\in M_1$.
	\end{itemize}	
	\end{claim}
	\begin{clproof}
	Let $Y$ be the set of all vertices in $A_2\setminus V(H)$ having a neighbor in $A_1\setminus V(R_1)$.
	Let $Z$ be the set of all vertices in $A_1\setminus V(R_1)\setminus (M\cap A_1)$ having a neighbor in $(V(R)\setminus V(R_1))\cap B$.
	See Figure~\ref{fig:mvc} for an illustration of $M', Y$ and $Z$.
	Let $M''$ be the set obtained from $M'\cup Y\cup Z$ by 
	removing all vertices $v\in M'\cap B$ such that all the neighbors of $v$ in $A_1\setminus V(R_1)\setminus (M\cap A_1)$ are contained in $Z$.	 

	We show that $M''$ is a vertex cover of $G_{A_1, A_2\cup B}-V(R_1)$.
	Suppose there is an edge $yz$ in $G_{A_1, A_2\cup B}-V(R_1)$ not covered by $M''$.
	As $Y$ hits all edges between $A_1$ and $A_2$ in $G_{A_1, A_2\cup B}-V(R_1)$, 
	this edge is an edge between $A_1$ and $B$. Assume that $y\in A_1$ and $z\in B$.
	
	As $V(R)\cap A_1\subseteq V(R_1)\cap A_1$ and $M\cap A_1=M'\cap A_1$,
	$z$ cannot be in $B\setminus (V(R)\cup M)$. 
	Thus, either $z\in (V(R)\setminus V(R_1))\cap B$ or $z\in (M\setminus M')\cap B$.
	Since $Z$ covers all edges between $A_1\setminus V(R_1)\setminus (M\cap A_1)$ and $(V(R)\setminus V(R_1))\cap B$, 
	$z$ is contained in $(M\setminus M')\cap B$.
	In this case, $z$ is a vertex in $M$ covering the edge $yz$ which is not covered by any other vertex in $M$, 
	and thus by definition of $M'$, 
	$M'$ includes $z$. Then one of $y$ and $z$ is contained in $(M'\cap B) \cup Z$. 
	This is a contradiction.
	Therefore, $M''$ is a vertex cover of $G_{A_1, A_2\cup B}-V(R_1)$.

	Now, we take a minimal vertex cover $M_1$ of $G_{A_1, A_2\cup B}-V(R_1)$ contained in $M''$.
	We have $V(R)\cap M_1=\emptyset$.
	Since each vertex of $M'\cap A$ covers some edge that is not covered by any other vertex in $M''$, 
	we have $M \cap A_1 = M'\cap A_1 \subseteq M_1$.
	Since every vertex in $Z$ meets some edge incident with $V(R)\setminus V(R_1)$, 
	$Z$ is contained in $M_1$.
	If $v$ is a vertex in $M\cap B$ incident with an edge $vw$ in $G_{A_1, B}-V(R)$ for some $w\notin V(R_1)$ that is not covered by any vertices in $M\setminus \{v\}$,
	then $v\in M'\cap B$.
	 By construction of $M''$, we have either $v\in M''\cap B$ or $w\in Z$.
	In particular if $w\notin Z$, then $v$ is the vertex covering the edge $vw$, and it also remains in $M_1$.
\end{clproof}
	
	By Claim \ref{claim:condition4}, the fourth and fifth condition of being a restriction are satisfied, so $(R_1, M_1)$ is a restriction of $(R, M)$. It remains to show that $F_1^*$ respects $(R_1, M_1)$. By construction, $R_1$ is the reduced forest of $F_1^*$ so we only have to show that that $V(F_1^*) \cap M_1 = \emptyset$, and in particular, by the construction given in the proof of Claim~\ref{claim:condition4}, it suffices to prove the following.
	\begin{claim}
		Let $Z$ be as in the proof of Claim~\ref{claim:condition4}. Then, $Z \cap V(F_1^*) = \emptyset$.
	\end{claim}
	\begin{clproof}
		Suppose not; let $x \in Z \cap V(F_1^*)$. By construction, $x \notin V(R_1)$ and $x$ has a neighbor $y$ in $(V(R) \setminus V(R_1)) \cap B$. Then, $x$ is either a leaf of $F_1^*$ or contained in a single-edge component of $F_1^*$: Since the edge $\{x, y\}$ is contained in $H$, it is also contained in $F_1^*$, so $x$ is not isolated in $F_1^*$. We can conclude that $y$ is the only neighbor of $x$ in $F_1^*$. However, neither $x$ nor $y$ is contained in $R_1$, a contradiction with the fact that $R_1$ is a reduced forest of $F_1^*$.
	\end{clproof}
	We can conclude that $F_1^*$ respects $(R_1, M_1)$.
	\end{proof}
	
	\begin{proposition}\label{prop:restriction2}
	Let $(A_1, A_2, B)$ be a vertex partition of a graph $G$.
	Let $R$ be an induced forest in $G_{A_1\cup A_2, B}$  and $M$ be a minimal vertex cover 
	of $G_{A_1\cup A_2, B}-V(R)$. 
	Let $H$ be an induced forest in $G[A_1\cup A_2\cup \bd(B)]-E(G[\bd(B)])$ respecting $(R, M)$ and 
	for each $i\in \{1, 2\}$, 
	\begin{itemize}
	\item let $(R_i, M_i)$ be a restriction of $(R, M)$ that
	$H\cap G[A_i\cup \bd(A_{3-i}\cup B)]-E(G[\bd(A_{3-i}\cup B)])$ respects (guaranteed by Proposition~\ref{prop:restriction1}), and 
	\item let $P_i$ be the partition of $\mathcal{C}(R_i)$ such that for $C, C'\in \mathcal{C}(R_i)$, 
	$C$ and $C'$ are in the same part if and only if they are contained in the same connected component of $H\cap G[A_i\cup \bd(A_{3-i}\cup B)]-E(G[\bd(A_{3-i}\cup B)])$.
	\end{itemize}
	Then $(R, R_1, R_2, P_1, P_2)$ is compatible.

	\end{proposition}
	\begin{proof}
	Let $Q$ be the auxiliary graph of  $(R, R_1, R_2, P_1, P_2)$.	
	It is not difficult to see that if $Q$ contains a cycle, then $H$ also contains a cycle, which leads to a contradiction.
	Thus, $Q$ has no cycles.
	\end{proof}

	\subsection{Bottom to top}\label{sec:reducedforest:bottom:top}
	
	\begin{proposition}\label{prop:correctness}
	Let $(A_1, A_2, B)$ be a vertex partition of a graph $G$.
	Let $R$ be an induced forest in $G_{A_1\cup A_2, B}$  and $M$ be a minimal vertex cover 
	of $G_{A_1\cup A_2, B}-V(R)$ such that for every vertex $x$ of degree at most $1$ in $R$,
	$\potentialleaves_{R,M}(x)\neq \emptyset$. 
	For each $i\in \{1, 2\}$, 
	\begin{itemize}
	\item let $R_i$ be an induced forest in $G_{A_i, A_{3-i}\cup B}$ and $M_i$ be a minimal vertex cover 
	of $G_{A_i, A_{3-i}\cup B}-V(R_i)$, and $H_i$ be an induced forest in $G[A_i\cup \bd(A_{3-i}\cup B)]-E(G[\bd(A_{3-i}\cup B)])$ respecting $(R_i, M_i)$, 
	\item let $P_i$ be the partition of $\mathcal{C}(R_i)$ such that for $C, C'\in \mathcal{C}(R_i)$, 
	$C$ and $C'$ are in the same  part if and only if they are contained in the same connected component of $H_i$,
	\item 	$(R_i, M_i)$ is a restriction of $(R, M)$.
	\end{itemize}
	Furthermore,
	\begin{itemize}
	\item every vertex in $(V(R) \setminus (V(R_1) \cup V(R_2))) \cap B$ has at least two neighbors in $(V(R_1)\cap A_1)\cup (V(R_2)\cap A_2)$,
	\item $(R, R_1, R_2, P_1, P_2)$ is compatible.
	\end{itemize}
	Then there is an induced forest $H$ in $G[A_1\cup A_2\cup \bd(B)]-E(G[\bd(B)])$ respecting $(R, M)$ such that
	\begin{itemize}
	\item 	$V(H)\cap (A_1\cup A_2)= (V(H_1)\cap A_1)\cup (V(H_2)\cap A_2)$.
	\end{itemize}
	\end{proposition}
	\begin{proof}
	As $(R, R_1, R_2, P_1, P_2)$ is compatible, we can verify that 
		\[ 
		H^* \defeq G[V(H_1)\cup V(H_2)\cup V(R)]
		\] 
	 is a forest.
	Let $H$ be the graph obtained from $H^*-(B\setminus V(R))$ by adding a potential leaf of each vertex in $V(R) \cap (A_1 \cup A_2)$ of degree at most $1$ in $R$ and removing all edges between vertices in $B$.
	We show that $H$ is a forest.
	
	\begin{claim}\label{claim:forest}
	$H$ is a forest such that $V(H)\cap (A_1\cup A_2)= (V(H_1)\cap A_1)\cup (V(H_2)\cap A_2)$.
	\end{claim}
	\begin{clproof}
	Since $H^*$ is a forest,
	$H^*-(B\setminus V(R))$ is also a forest.
	Adding a potential leaf of a vertex in $V(R)\cap (A_1\cup A_2)$ preserves the property of being a forest, as we removed all edges in $B$.
	When we take  $H$ from $H^*$, we only change the vertices in $B$.
	Therefore, we have $V(H)\cap (A_1\cup A_2)= (V(H_1)\cap A_1)\cup (V(H_2)\cap A_2)$.
	\end{clproof}
	
	In the remainder, we prove that $H$ respects $(R, M)$.
	We need to verify that 
		\begin{enumerate}[(i)]
	\item $R$ is a reduced forest of $G_{A_1\cup A_2,B}\cap H$. \label{prop:correctness:respects:1}
	\item $V(H) \cap M = \emptyset$. \label{prop:correctness:respects:2}
\end{enumerate}
	Condition (\ref{prop:correctness:respects:2}) is easy to verify: 
	Since we remove all vertices in $M$ when we construct $H$ from $G[V(H_1)\cup V(H_2)\cup V(R)]$, 
	we have $V(H)\cap M=\emptyset$. We now verify condition (\ref{prop:correctness:respects:1}).
	
\newcommand{\hnew}{H_{A, B}}	
	
	Let $\hnew\defeq H\cap G_{A_1\cup A_2, B}$.
	We first verify 
	\begin{claim}\label{claim:degree1}
	Every vertex of $V(\hnew)\setminus V(R)$ has degree at most $1$ in $\hnew$.
	\end{claim}
	\begin{clproof}
	Let $v\in V(\hnew)\setminus V(R)$. First assume that $v\in A_1\cup A_2$. Without loss of generality, we assume $v\in A_1$.
	Since $v$ is not contained in $R$ and $v\notin M\cap A_1\subseteq M_1$, its neighborhood in $G_{A_1, B}-V(R)$ is contained in $M$.	
	As $\hnew$ does not contain a vertex in $M$, the neighborhood of $v$ in $\hnew$ is contained in $V(R)\cap B$.

	Suppose for contradiction that $v$ has at least two neighbors in $V(R)\cap B$.
	Since $(R_1, M_1)$ is a restriction of $(R,M)$, by the third condition of the statement of the proposition, $v$ is also not contained in $R_1$.
	If $v$ has at least two neighbors in $V(R_1)\cap B$, then $v$ should be contained in $R_1$, a contradiction.
	Therefore, $v$ has at least one neighbor in $(V(R)\setminus V(R_1))\cap B$, say $w$.
	Then $vw$ is an edge in $G_{A_1, A_2\cup B}-V(R_1)$, 
	so $M_1$ contains $v$ or $w$.
	Since $v\in V(H_1)$, $v$ is not contained in $M_1$, and thus $w\in M_1$.
	But this contradicts the assumption that $M_1\cap V(R)=\emptyset$, which is the forth condition of being a restriction.
	Therefore, $v$ has at most one neighbor in $V(R)\cap B$, as required.

	Now we assume $v\in B$. By construction, $v$ is a potential leaf of some vertex in $R$. Thus $v$ has degree at most $1$ in $\hnew$, as required.
	\end{clproof}

	We argue that we can take $R$ as a reduced forest of $\hnew$. 
	Let $v\in V(R)$. If $v$ has degree at least $2$ in $\hnew$, then 
	$v$ is contained in any reduced forest of $\hnew$.
	Suppose $v$ has degree at most $1$ in $\hnew$.
	Suppose $v\in A_1\cup A_2$. In this case, by construction, $v$ is incident with its potential leaf in $\hnew$, say $w$.
	This means that $vw$ is a single-edge component in $\hnew$, and we can take $v$ as a vertex in $R$.
	
	Now, suppose $v\in B$.
	First assume that $v\in V(R_i)$ for some $i\in \{1,2\}$. 
	If $v$ has a neighbor in $R_i$, then it also has at least one potential leaf in $H_i\cap G_{A_i, A_{3-i}\cup B}$, 
	and thus $v$ has degree $2$ in $\hnew$, a contradiction.
	Thus, $v$ has no neighbor in $R_i$, and has exactly one potential leaf, say $w$.
	By Claim~\ref{claim:degree1}, $v$ is the unique neighbor of $w$ in $R$, 
	and thus $vw$ is a single-edge component of $\hnew$. Thus, we can take $v$ as a vertex in $R$.
	Suppose $v\in (V(R) \setminus (V(R_1) \cup V(R_2))) \cap B$.
	Then by the precondition, it has at least two neighbors in $(V(R_1)\cap A_1)\cup (V(R_2)\cap A_2)\subseteq (V(H_1)\cap A_1)\cup (V(H_2)\cap A_2)$.
	Therefore, it is contained in any reduced forest of $\hnew$.
	It shows that $R$ is a reduced forest of $\hnew$.
	
	Note that for each $i\in \{1, 2\}$, $V(H_i)\cap A_i$ avoids $M\cap A_i$.
	Furthermore, when we construct $\hnew$, we removed all vertices in $M\cap B$.
	Therefore, we have $V(\hnew)\cap M=\emptyset$, as required.	
	\end{proof}

\section{Feedback Vertex Set on graphs of bounded mim-width}\label{sec:fvs}
	
In this section we give an algorithm that solves the \textsc{Feedback Vertex Set} problem on graphs on $n$ vertices together with a branch decomposition of mim-width $w$ in time $n^{\cO(w)}$.

First, we observe that given a graph $G$, a subset of its vertices $S \subseteq V(G)$ is by definition a feedback vertex set if and only if $G - S$, the induced subgraph of $G$ on vertices $V(G) \setminus S$, is an induced forest. It is therefore readily seen that computing the minimum size of a feedback vertex set is equivalent to computing the maximum size of an induced forest, so in the remainder of this section we solve the following problem which is more convenient for our exposition.

\parproblemdef
	{Maximum Induced Forest/Mim-Width}
	{A graph $G$ on $n$ vertices, a branch decomposition $(T, \decf)$ of $G$ and an integer $k$.}
	{$w \defeq \mimw(T, \decf)$.}
	{Does $G$ contain an induced forest of size at least $n - k$?}

	We furthermore assume that $G$ is connected; otherwise, we can solve it for each connected component.
	Also, we assume $G$ contains at least $2$ vertices.
	
We solve the \textsc{Maximum Induced Forest} problem by bottom-up dynamic programming over $(T, \decf)$, the given branch decomposition of $G$, starting at the leaves of $T$. Let $t \in V(T)$ be a node of $T$. To motivate the table indices of the dynamic programming table, we now observe how a solution to \textsc{Maximum Induced Forest}, an induced forest $\cF$, interacts with the graph $G_{t+\bd} \defeq G[V_t \cup \bd(\bar{V_t})]-E(G[\bd(\bar{V_t})])$. The intersection of $\cF$ with $G_{t+\bd}$ is an induced forest which throughout the following we denote by $\cF_{t+\bd} \defeq \cF[V(G_{t+\bd})]$.
Since we want to bound the number of table entries by $n^{\cO(w)}$, we have to focus in particular on the interaction of $\cF$ with the crossing graph $\crosg{t}$ which is an induced forest in $\crosg{t}$, denoted by $\cF_{t, \bar{t}} \defeq \cF[V(\crosg{t})]$.

However, it is not possible to enumerate all induced forests in a crossing graph as potential table indices: Consider for example a star on $n$ vertices and the cut consisting of the central vertex on one side and the remaining vertices on the other side. This cut has mim-value $1$ but it contains $2^n$ induced forests, since each vertex subset of the star induces a forest on the cut. The remedy for this issue are {\em reduced} (induced) forests, introduced in Section \ref{sec:reducedforest}. 

In particular, at each node $t \in V(T)$, we only consider reduced forests as possible (parts of) indices for the table entries, and by Lemma \ref{lem:reducedforest}, the number of reduced forests in each cut of mim-value $w$ is bounded by $\cO(n^{6w})$. We now analyze the structure of $\cF_{t, \bar{t}}$ to motivate the objects that can be used to represent $\cF_{t, \bar{t}}$ in such a way that the number of all possible table entries remains bounded by $n^{\cO(w)}$.

The induced forest $\cF_{t, \bar{t}}$ has three types of vertices in $G_{t, \bar{t}}$:
\begin{itemize}
	\item The vertices of the reduced forest $\reduced(\cF_{t, \bar{t}})$ of $\cF_{t, \bar{t}}$.\label{mif:solution:reduced:forest}
	\item The leaves of the induced forest $\cF_{t, \bar{t}}$, denoted by $\leaves(\cF_{t, \bar{t}})$.\label{mif:solution:leaves}
	\item Vertices in $\cF_{t, \bar{t}}$ that do not have a neighbor in $\cF_{t, \bar{t}}$ on the opposite side of the boundary, in the following called {\em non-crossing} vertices and denoted by $\noncrossing(\cF_{t, \bar{t}})$.\label{mif:solution:noncrossing}
\end{itemize}

As outlined above, the only type of vertices in $\cF_{t, \bar{t}}$ that will be used as part of the table indices are the vertices of a reduced forest of $\cF_{t, \bar{t}}$, since otherwise, the number of possible indices might be exponential in $n$. Hence, we neither know about the leaves of $\cF_{t, \bar{t}}$ nor its non-crossing vertices upon inspecting this part of the index. Suppose we have a vertex $v \in (\leaves(\cF_{t, \bar{t}}) \cup \noncrossing(\cF_{t, \bar{t}})) \cap V_t$ and consider $N_{\bar{t}}^*(v) \defeq (N(v) \cap \bar{V_t}) \setminus V(\reduced(\cF_{t, \bar{t}}))$. Then, $\cF_{t, \bar{t}}$ does not use any vertex in $x \in N_{\bar{t}}^*(v)$: If $v$ is a leaf in $\cF_{t, \bar{t}}$, then the presence of the edge $\{v, x\}$ would make it a non-leaf vertex and if $v$ is a non-crossing vertex, the presence of $\{v, x\}$ would make $v$ a vertex incident to an edge of the forest crossing the cut.  An analogous point can be made for a vertex in $(\leaves(\cF_{t, \bar{t}}) \cup \noncrossing(\cF_{t, \bar{t}})) \cap \bar{V_t}$.  In the table indices, we capture this property of $\cF_{t, \bar{t}}$ by considering a minimal vertex cover of $G_{t, \bar{t}} - V(\reduced(\cF_{t, \bar{t}}))$ that avoids all leaves and non-crossing vertices of $\cF_{t, \bar{t}}$. We observe that such a minimal vertex cover always exists. (Note that $\leaves(\cF_{t, \bar{t}}) \cup \noncrossing(\cF_{t, \bar{t}})$ is an independent set in $G_{t, \bar{t}}$.)

\begin{observation}\label{obser:is:avoid:mvc}
	Let $G$ be a graph and $X \subseteq V(G)$ an independent set in $G$. Then, there exists a minimal vertex cover $M$ of $G$ such that $X \cap M = \emptyset$.
\end{observation}

\newcommand{\mIn}[1]{M_{#1}^{\mathrm{in}}}
\newcommand{\mOut}[1]{M_{#1}^{\mathrm{out}}}

Lastly, we have to keep track of how the connected components of $\cF_{t, \bar{t}}$ (respectively, $\reduced(\cF_{t, \bar{t}})$) are joined together via the forest $\cF_{t+\bd}$.
This forest induces a partition of $\cC(\reduced(\cF_{t, \bar{t}}))$ in the following way: Two components $C_1, C_2 \in \cC(\reduced(\cF_{t, \bar{t}}))$ are in the same part of the partition if and only if  $C_1$ and $C_2$ are contained in the same connected component of $\cF_{t + \bd}$. 

We are now ready to define the indices of the dynamic programming table $\dptable$ to keep track of sufficiently much information about the partial solutions in the graph $G_{t+\bd}$. Throughout the following, we denote by $\cR_t$ the set of all induced forests of $\crosg{t}$ on at most $6w$ vertices (which by Lemma \ref{lem:reducedforest} contains all reduced forests in $\crosg{t}$). For $R \in \cR_t$, we let $\cM_{t, R}$ be the set of all minimal vertex covers of $\crosg{t} - V(R)$ and $\cP_{t, R}$ the set of all partitions of the connected components of $R$. 

For an illustration of the above discussion and also the definition of the table indices, which we start on now, see Figure \ref{fig:fvs:table:entries}.
For $(R, M, P) \in \cR_t \times \cM_{t, R} \times \cP_{t, R}$ and $i \in \{0,\ldots, n\}$, we set $\dptable[t, (R, M, P), i] \defeq 1$ (and to $0$ otherwise), if and only if the following conditions are satisfied. 

\begin{figure}
	\centering
		\centering
		\includegraphics[width=.70\textwidth]{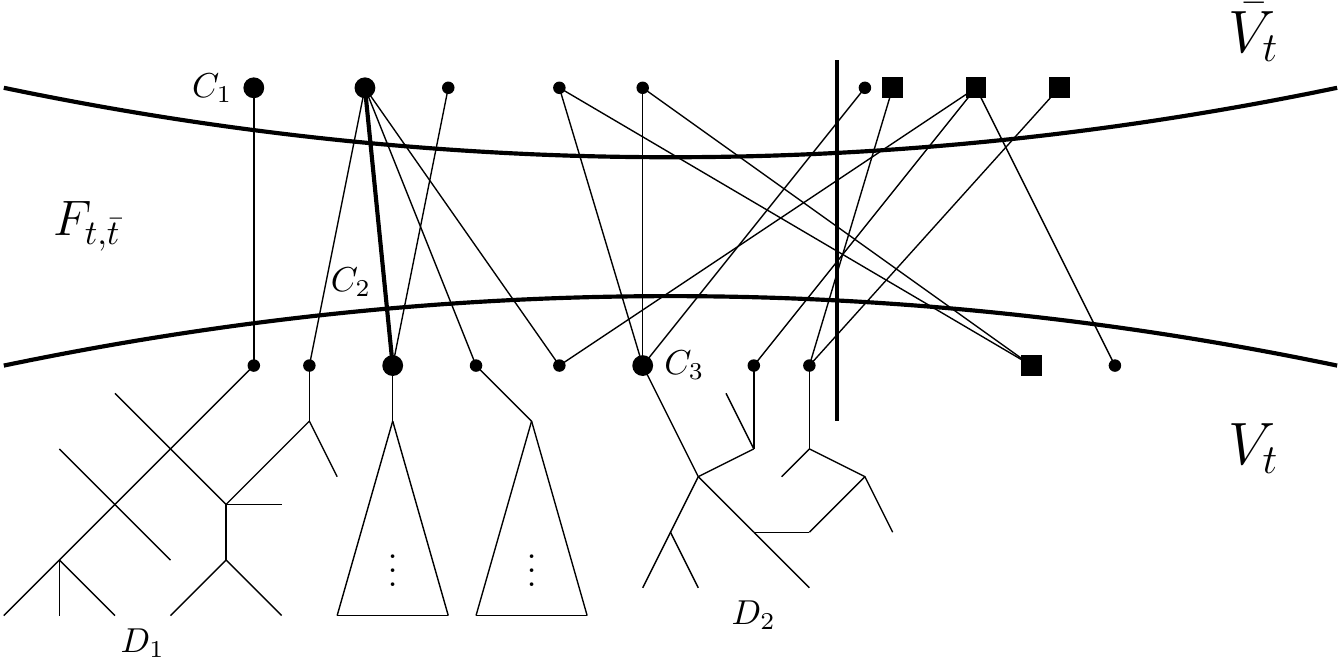}
	\caption{An example of a crossing graph $\crosg{t}$ together with an induced forest $\cF$ and their interaction. The forest $\cF_{t, \bar{t}} = \cF[V(\crosg{t})]$ is displayed to the left of the dividing line in the drawing and the 4 vertices and 1 edge in bold form a reduced forest $R$ of $\cF_{t, \bar{t}}$. The square vertices form a minimal vertex cover of $\crosg{t} - V(R)$ satisfying (\ref{fvs:table:index:3}). Furthermore, $C_i$ ($i \in [3]$) are the connected components of $R$ and $D_i$ ($i \in [2]$) are the connected components of $\cF$.}
		\label{fig:fvs:table:entries}
\end{figure}

\newcommand{\trimmed}[1]{#1^*}

\begin{enumerate}[(i)]
	\item There is an induced forest $F$ in $G[V_t \cup \bd(\bar{V_t})]-E(G[\bd(\bar{V_t})])$, such that $V(F)\cap V_t$ has size $i$.\label{fvs:table:index:1}
	\item Let $F_{t, \bar{t}} = F\cap G_{t, \bar{t}}$, i.e.\ $F_{t, \bar{t}}$ is the subforest of $F$ induced by the vertices of the crossing graph $G_{t, \bar{t}}$. Then, $R = \reduced(F_{t, \bar{t}})$, meaning that $R$ is a reduced forest of $F_{t, \bar{t}}$.\label{fvs:table:index:2}
	\item $M$ is a minimal vertex cover of $G_{t, \bar{t}} - V(R)$ such that $V(F) \cap M = \emptyset$.\label{fvs:table:index:3}
	\item $P$ is a partition of $\cC(R)$ such that 
	two components $C_1, C_2 \in \cC(R)$ are in the same part of the partition if and only if  $C_1$ and $C_2$ are contained in the same connected component of $F$.\label{fvs:table:index:4} 
\end{enumerate}

Regarding (\ref{fvs:table:index:3}), recall that even though the leaves and non-crossing vertices of $F_{t, \bar{t}}$ are still contained in $G_{t, \bar{t}} - V(R)$, a minimal vertex cover that avoids the leaves and non-crossing vertices of $F_{t, \bar{t}}$ always exists by Observation~\ref{obser:is:avoid:mvc}.

Recall that $r \in V(T)$ denotes the root of $T$, the tree of the given branch decomposition of $G$. From Property (\ref{fvs:table:index:1}) we immediately observe that the table entries store enough information to obtain a solution to \textsc{Maximum Induced Forest} after all table entries have been filled. In particular, we make

\begin{observation}
	$G$ contains an induced forest of size $i$ if and only if $\dptable[r, (\emptyset, \emptyset, \emptyset), i] = 1$.
\end{observation}

Before we proceed with the description of the algorithm, we first show that the number of table entries is bounded by a polynomial whose degree is linear in the mim-width $w$ of the given branch decomposition.

\begin{proposition}\label{prop:indices}
	There are at most $n^{\cO(w)}$ table entries in $\dptable$.
\end{proposition}
\begin{proof}
	Let $t \in V(T)$. We show that the number of table entries in $\dptable_t$ is bounded by $n^{\cO(w)}$ which together with the observation that $\card{V(T)} = \cO(n)$ yields the proposition. By definition, $\card{\cR_t} = \cO(n^{6w})$ and by the Minimal Vertex Covers Lemma we have for each $R \in \cR_t$ that $\card{\cM_{t, R}} = n^{\cO(w)}$. The size of $\cP_{t, R}$ is at most the number of partitions of a set of size $6w$, and hence at most $B_{6w} < (w/\log(w))^{\cO(w)}$ by standard upper bounds on the Bell number $B_{6w}$ (see e.g.~\cite{BT10}). Finally, there are $n+1$ choices for the integer $i$. To summarize, there are at most 
	\begin{align*}
		\cO\left(n^{6w}\right) \cdot n^{\cO(w)} \cdot (w/\log(w))^{\cO(w)} \cdot (n+1) = n^{\cO(w)}
	\end{align*}
table entries in $\dptable_t$ and the proposition follows.
\end{proof}

We now show how to compute the table entries in $\dptable$. First, we explain how to compute the entries in $\dptable_\ell$ for the leaves of $T$ and then how to compute the entries in the internal nodes of $T$ from the entries stored in the tables corresponding to their children.

\medskip
\noindent{\bf Leaves of $T$.} Let $t \in V(T)$ be a leaf of $T$ and $v = \decf^{-1}(t)$. Clearly, the crossing graph $\crosg{t}$ is a star $S$ with central vertex $v$ or a single edge.
Hence, any induced forest $F$ in $G[\{v\} \cup N(v)] - E(G[N(v)])$ satisfies that
 either $V(F)=\{v\}$ or $V(F)\subseteq N(v)$ or $F$ contains an edge in $G_{t, \bar{t}}$.
 In the last case, either $F$ is a single edge or a star with central vertex $v$.
 Let $R$ be a reduced forest of $F$.
 
If $V(F)=\{v\}$, then $R=\emptyset$, $M=N(v)$, $P=\emptyset$, and $i=1$.
If $V(F)\subseteq N(v)$, then $R=\emptyset$, $M=\{v\}$, $p=\emptyset$, and $i=0$.
Throughout the following, we assume $F$ contains an edge in $G_{t, \bar{t}}$. 

Suppose $F$ is a single edge $\{v, w\}$. Then, $R$ is either the vertex $v$ or the vertex $w$. If $V(R) = \{v\}$, then $G_{t, \bar{t}} - V(R)$ does not contain any edges and hence $\cM_{t, R} = \{\emptyset\}$. Furthermore, $F$ has size one in $G[V_t] = G[\{v\}]$.
If $V(R)=\{w\}$, then $v$ is a leaf in $F$ and hence the only minimal vertex cover satisfying (\ref{fvs:table:index:3}) is the set of neighbors of $v$ without $w$, i.e.\ the set $N(v) \setminus \{w\}$.
The size of $F$ in $G[V_t]$ is $1$.
In both cases, $F$ only has one component, so $\cP_{t, R} = \{\{R\}\}$.

Now suppose that $F$ has at least three vertices. Then, $F$ is a star with central vertex $v$ and hence, the reduced forest of any such $F$ is the single vertex $v$.
Since the vertices of $F$ in $\bar{V_t}$ are not counted in the table entry by (\ref{fvs:table:index:1}), we only have to consider one index where the reduced forest is $v$, the minimal vertex cover is empty (again since $G_{t, \bar{t}} - \{v\}$ does not have any edges), the partition of $R$ is the singleton partition and $i = 1$, since $F$ has size one in $G[V_t] = G[\{v\}]$.

We furthermore have to represent the empty solution, i.e.\ the case when $F = \emptyset$. Then, both $\{v\}$ and $N(v)$ are feasible minimal vertex covers and clearly, $P = \emptyset$. To summarize, the table entries for the leaf $t$ are set as follows.
\begin{align*}
	\dptable[t, (R, M, P), i] \defeq \left\{\begin{array}{ll}
		1, &\mbox{if } R = \emptyset, M=N(v), P = \emptyset, i = 1 \\
		1, &\mbox{if } R = \emptyset, M \in \{\{v\}, N(v)\}, P = \emptyset, i = 0 \\
		1, &\mbox{if } R = G[\{v\}], M = \emptyset, P = \{R\}, i = 1 \\
		1, &\mbox{if } R = G[\{w\}] \mbox{ where } w \in N(v), M = N(v) \setminus \{w\}, P = \{R\}, i = 1 \\
		0, &\mbox{otherwise}
	\end{array}	
	\right.
\end{align*}

\newcommand\dpindex{\mathfrak{I}}

\noindent{\bf Internal Nodes of $T$.} Let $t \in V(T)$ be an internal node with children $a$ and $b$. 
	Using Propositions~\ref{prop:restriction1}, \ref{prop:restriction2} and \ref{prop:correctness}, 
	we can show the following.
	
	\begin{proposition}\label{claim:correctness}
	Let $\dpindex = [(R, M, P), i] \in 
\left(\cR_t \times \cM_{t, R_t} \times \cP_{t, R_t}\right) \times \{0,\ldots,n\}$ such that 
	for every vertex $x$ of degree at most $1$ in $R$,
	$\potentialleaves_{R,M}(x)\neq \emptyset$. 
	Then $\dptable[t, (R, M, P), i]=1$ if and only if 
	there are restrictions $(R_a, M_a)$ and $(R_b, M_b)$ of $(R, M)$ to $G_{a, \bar{a}}$ and $G_{b, \bar{b}}$, respectively, 
	and partitions $P_a$ and $P_b$ of $\mathcal{C}(R_a)$ and $\mathcal{C}(R_b)$, respectively, 
	and integers $i_a$ and $i_b$ such that 
	\begin{itemize}
	\item $\dptable[t_a, (R_a, M_a, P_a), i_a]=1$ and $\dptable[t_b, (R_b, M_b, P_b), i_b]=1$,
	\item $(R, R_a, R_b, P_a, P_b)$ is compatible and $P=\mathcal{U}(R, R_1, R_2, P_1, P_2)$,
	\item every vertex in $(V(R) \setminus (V(R_1) \cup V(R_2))) \cap B$ has at least two neighbors in $(V(R_1)\cap A_1)\cup (V(R_2)\cap A_2)$,
	\item $i_a+i_b=i$.
	\end{itemize}
	\end{proposition}
	\begin{proof}
	Suppose $\dptable[t, (R, M, P), i]=1$.
	Let $H$ be an induced forest of $G[V_t\cup \bd(\bar{V_t})]-E(G[\bd(\bar{V_t})])$ that is a partial solution with respect to $(R, M, P)$ and $i$.
	For each $x\in \{a,b\}$, let $H_x\defeq H\cap (G[V_x\cup \bd(\bar{V_x})]-E(G[\bd(\bar{V_x})]))$.
	By Proposition~\ref{prop:restriction1}, 
	there are restrictions $(R_a, M_a)$ and $(R_b, M_b)$ of $(R,M)$ to $V_a$ and $V_b$, respectively, 
	such that 
	\begin{itemize}
	\item $H_a$ respects $(R_a, M_a)$, 
	and $H_b$ respects $(R_b, M_b)$, and 
	\item every vertex in $(V(R) \setminus (V(R_1) \cup V(R_2))) \cap B$ has at least two neighbors in $(V(R_1)\cap A_1)\cup (V(R_2)\cap A_2)$.
	\end{itemize}
	For each $x\in \{a,b\}$, let $P_x$ be the partition of $\mathcal{C}(R_x)$ such that 
	$C_1$ and $C_2$ in $\mathcal{C}(R_x)$ are contained in the same part if and only if
	they are contained in the same connected component of $H_x$.
	Then by Proposition~\ref{prop:restriction2}, the tuple $(R, R_a, R_b, P_a, P_b)$ is compatible and it is not difficult to verify that $P=\mathcal{U}(R, R_1, R_2, P_1, P_2)$.
	Let $i_x\defeq \abs{V(H)\cap V(G[V_x])}$. Then, $i_a+i_b=i$ as $V_a$ and $V_b$ are disjoint.
	This concludes the forward direction.
	
	To verify the converse direction, suppose the latter conditions hold.
	For each $x\in \{a,b\}$, let $H_x$ be an induced forest in $G[V_x\cup \bd(\bar{V_x})]-E(G[\bd(\bar{V_x})])$ that is a partial solution with respect to $(R_x, M_x, P_x)$ and $i_x$.
	By Proposition~\ref{prop:correctness}, 
	there is an induced forest $H$ in $G[V_t \cup \bd(\bar{V_t})]-E(G[\bd(\bar{V_t})])$ respecting $(R, M)$ such that 
	\[
		H\cap G[V_t]=(H_a\cap G[V_a])\cup (H_b\cap G[V_b]).
	\]
	Therefore, we have $\abs{V(H)\cap V_t}=\abs{V(H_a)\cap V_a}+\abs{V(H_b)\cap V_b}=i_a+i_b=i$, so $\dptable[t, (R, M, P), i]=1$, as required.
		\end{proof}

	Based on Proposition~\ref{claim:correctness}, we can proceed with the computation of the table at an internal node $t$ with children $a$ and $b$. 
   Let $\dpindex = [(R, M, P), i] \in \left(\cR_t \times \cM_{t, R_t} \times \cP_{t, R_t}\right) \times \{0,\ldots,n\}$.

\begin{description}
	\item[Step 1 (Valid Index).]  We verify whether $\dpindex$ is valid, i.e.\ whether it can represent a valid partial solution in the sense of the definition of the table entries. That is, each vertex of degree at most $1$ in $R$ has to have at least one potential leaf. 

	\item[Step 2 (Reduced Forests).] 
	We consider all pairs of indices for $\dptable_a$ and $\dptable_b$ denoted by
\begin{itemize}
	\item $\dpindex_a = [(R_a, M_a, P_a), i_a] \in \left(\cR_a \times \cM_{a, R_a} \times \cP_{a, R_a} \right) \times \{0,\ldots,n\}$ and
	\item $\dpindex_b = [(R_b, M_b, P_b), i_b] \in \left(\cR_b \times \cM_{b, R_b} \times \cP_{b, R_b}\right) \times \{0,\ldots,n\}$.
\end{itemize}
	We check
		\begin{itemize}
	\item $(R_a, M_a)$ and $(R_b, M_b)$ are restrictions of $(R, M)$ to $G_{a, \bar{a}}$ and $G_{b, \bar{b}}$ respectively,
	\item $\dptable[t_a, (R_a, M_a, P_a), i_a]=1$ and $\dptable[t_b, (R_b, M_b, P_b), i_b]=1$,
	\item $(R, R_a, R_b, P_a, P_b)$ is compatible  and $P=\mathcal{U}(R, R_1, R_2, P_1, P_2)$,
	\item every vertex in $(V(R) \setminus (V(R_1) \cup V(R_2))) \cap B$ has at least two neighbors in $(V(R_1)\cap A_1)\cup (V(R_2)\cap A_2)$,
	\item $i_a+i_b=i$.
	\end{itemize}
	If there are $\dpindex_a$ and $\dpindex_b$ satisfying all of the above conditions, then we assign $\dptable[t, (R, M, P), i]=1$ and otherwise, we assign $\dptable[t, (R, M, P), i]=0$.
	Correctness follows from Proposition~\ref{claim:correctness}.
\end{description}

We finish by analyzing the running time of the algorithm.
At each node $t \in V(T)$, we can enumerate all table indices in time $n^{\mathcal{O}(w)}$ by Corollary~\ref{cor:mvc:lemma} and Proposition~\ref{prop:indices}.
	Let $\dpindex = [(R, M, P), i] \in \left(\cR_t \times \cM_{t, R_t} \times \cP_{t, R_t}\right) \times \{0,\ldots,n\}$.
   	If $t$ is a leaf node, then $\dptable[t, (R, M, P), i]$ can be computed in linear time. 
	Assume that $t$ is an internal node.
	We can check in linear time whether $\dpindex$ is valid or not.
	Next, for all pairs of $\dpindex_a= [(R_a, M_a, P_a), i_a] \in \left(\cR_a \times \cM_{a, R_a} \times \cP_{a, R_a} \right) \times \{0,\ldots,n\}$ and $\dpindex_b= [(R_b, M_b, P_b), i_b] \in \left(\cR_b \times \cM_{b, R_b} \times \cP_{b, R_b}\right) \times \{0,\ldots,n\}$
	we verify the conditions of Step 2 hold, which can be done in time $\cO(n^2)$.
	Therefore, by Proposition~\ref{prop:indices}, we can decide whether $\dptable[t, (R, M, P), i]=1$ or not in time $n^{\mathcal{O}(w)}$.
	As $T$ contains $\mathcal{O}(n)$ nodes, we can solve \textsc{Feedback Vertex Set} in time $n^{\mathcal{O}(w)}$.
	
	We can easily modify our algorithm into an algorithm solving the weighted version of the problem.
	In \textsc{Weighted Feedback Vertex Set}, we are given a graph and a weight function $\omega:V(G)\rightarrow \mathbb{R}$, 
	we want to find a set $S$ with minimum $\omega(S)$ such that $G-S$ has no cycles.
	Similar to \textsc{Feedback Vertex Set}, we can instead solve the problem of finding an induced forest $F$ with maximum $\omega(V(F))$.
	Instead of specifying $i$ in the table index $[t, (R, M, P), i]$,
	we store at $\dptable[t, (R, M, P)]$ the maximum value $\omega(V(F) \cap V_t)$ over all induced forests $F$ that respect $(R, M)$ and whose connectivity partition is $P$.
	The procedure for leaf nodes is analogous. In the internal node, 
	we compare all pairs $(R_a, M_a, P_a)$ and $(R_b, M_b, P_b)$ for children $t_a$ and $t_b$, and take the maximum among all sums $\mathcal{T}[t_a, (R_a, M_a, P_a)]+\mathcal{T}[t_b, (R_b, M_b, P_b)]$.
	Therefore, we can solve \textsc{Weighted Feedback Vertex Set} in time $n^{\mathcal{O}(w)}$ as well. We have proved Theorem \ref{thm:FVSmim}.

\section{Hamiltonian Cycle for linear mim-width 1}\label{sec:hamcyc}

\begin{theorem}
    \textsc{Hamiltonian Cycle} is $\NP$-complete on graphs of linear mim-width $1$, even if given the mim-width decomposition.
    \end{theorem}
    \begin{proof}
    Itai et al \cite{IPS82} showed that given a bipartite graph $G$ with maximum degree 3, it is $\NP$-complete to decide if it has a Hamiltonian cycle, while Panda and Pradhan \cite{PP08} construct, from this graph $G$, a rooted directed path graph $H$ such that $H$ has a Hamiltonian cycle if and only if $G$ does. Here we show that the construction of \cite{PP08} can be used to also output a linear mim-width 1 decomposition of $H$, in polynomial time, which will prove our result.

    Let $G$ be the bipartite graph on bipartition $(\{v_1, v_2, \ldots, v_m\}, \{w_1, w_2, \ldots, w_m\})$
    that has maximum degree at most $3$, and has no leaves.
    Let us consider the construction, given by Panda and Pradhan \cite{PP08}, of a new graph $H$ from $G$ as follows:
    \begin{enumerate}
    \item For each $i\in \{1, \ldots, m\}$, we introduce vertices $X_i, Y_i$ to $H$, and if $w_i$ has degree $3$,
    then we introduce a vertex $Z_i$ additionally. We let $\cX_{\ge i} \defeq \bigcup_{i' \ge i} X_i$, $\cX \defeq \cX_{\ge 1}$, $\cY_{\ge i} \defeq \bigcup_{i' \ge i} Y_i$, $\cY \defeq \cY_{\ge 1}$, $\cZ_{\ge i} \defeq \bigcup_{\substack{i' \ge i \\ \deg(w_i)=3}} Z_i$ and $\cZ \defeq \cZ_{\ge 1}$.
    \item For each $v_iw_j\in E(G)$, we introduce a vertex $A_{i,j}$ to $H$. We let $\cA_j \defeq \bigcup_{i\colon v_iw_j \in E(G)} A_{i, j}$ and $\cA \defeq \bigcup_{j \le m} \cA_j$.
    \item For each $i\in \{1, \ldots, m\}$, we make $\{X_i\}\cup \{A_{i',j}:i'\le i, v_{i'}w_j\in E(G)\}$ a clique in $H$.
    \item For each $j\in \{1, \ldots, m\}$, we make $\{Y_j\}\cup \{A_{i,j}:v_iw_j\in E(G)\}$ a clique in $H$,
    and $w_i$ has degree $3$, then we also make $\{Z_j\}\cup \{A_{i,j}:v_iw_j\in E(G)\}$ a clique in $H$.
    \end{enumerate}
    That concludes the construction of the graph $H$.
	We now summarize the most important properties of $H$ which will be useful later in the proof.
	\begin{enumerate}
		\makeatletter	
		\renewcommand\labelenumi{(H\arabic{enumi})}
		\renewcommand\theenumi{H\arabic{enumi}}
		\makeatother
		\item $\cX \cup \cY \cup \cZ$ is an independent set in $H$. \label{proof:hamcyc:properties:1} 
		\item $\cA$ is a clique in $H$. \label{proof:hamcyc:properties:2}
		\item For $j \in [m]$ and $v \in \{Y_j, Z_j\}$, $N(v) = \cA_j$. \label{proof:hamcyc:properties:3}
		\item For $i_1, i_2 \in [m]$ with $i_1 \le i_2$, $N(X_{i_1}) \subseteq N(X_{i_2})$ and $N(A_{i_1, \cdot}) \cap \cX \subseteq N(A_{i_2, \cdot}) \cap \cX$. \label{proof:hamcyc:properties:4}
		\item For $j \in [m-1]$, no vertex in $\cA_j$ has a neighbor in $\cY_{\ge j+1} \cup \cZ_{\ge j+1}$. \label{proof:hamcyc:properties:5}
	\end{enumerate}   
    
        Formally, a branch decomposition $(T, \decf)$ of $H$ is linear if $T$ consists of a path on $|V(H)|-2$ nodes with a leaf added to each inner node of the path, with $\decf$ a bijection between the leaves of $T$ and the vertices of $H$. For simplicity, let us simply say that a linear branch decomposition is a total ordering of the vertices of $H$.
    For each $j\in \{1, \ldots, m\}$, let $L_j$ be any ordering of vertices in $\cA_j$,
    and let $U_j$ be the ordering $(Y_j, Z_j)$ if $w_j$ has degree $3$, and $(Y_j)$ otherwise.
    We claim that the linear branch decomposition\footnote{For two disjoint total orderings $X$ and $Y$ we define their sum $X \oplus Y$ as follows. Suppose $X = x_1, x_2, \ldots, x_r$ and $Y = y_1, y_2, \ldots, y_s$, then $X \oplus Y = x_1, x_2, \ldots, x_r, y_1, y_2, \ldots y_s$.}
    \begin{align}
    	U_1\oplus L_1\oplus U_2\oplus L_2 \oplus \cdots \oplus U_m\oplus L_m \oplus (X_m, X_{m-1}, \ldots, X_1)\label{proof:hamcyc:branch:dec}
    \end{align}
    of $H$ has mim-width at most $1$.
    Let $v\in V(H)$, and let $S_v$ be the union of $\{v\}$ and the set of vertices appearing before $v$ in the ordering, and let $T_v\defeq V(H)\setminus S_v$.
    We divide into three cases depending on the where the vertex $v$ appears in the linear branch decomposition. Suppose for a contradiction that there exist $a_1, a_2 \in S_v$, $b_1, b_2 \in T_v$ such that $a_1b_1$, $a_2b_2 \in E(H)$ but $a_1b_2, a_2b_1 \notin E(H)$ (which would imply that the linear branch decomposition (\ref{proof:hamcyc:branch:dec}) has mim-width at least two).

	\begin{description}
		\item[Case 1 ($v \in U_k$ for some $k$).] By (\ref{proof:hamcyc:properties:3}), every vertex in $U_t$ where $t < k$, has no neighbors in $T_v$. Let $x \in T_v$ be a neighbor of $v$. By (\ref{proof:hamcyc:properties:1}), $x \in \cA$ so by (\ref{proof:hamcyc:properties:2}), $x$ is adjacent to every vertex in $L_1 \cup L_2 \cup \cdots \cup L_k$. We can conclude that $v$ is neither $a_1$ nor $a_2$. We have argued that $a_1, a_2 \in L_1 \cup \cdots \cup L_{k-1}$. By (\ref{proof:hamcyc:properties:4}), either $N(a_1) \cap \cX \subseteq N(a_2) \cap \cX$ or $N(a_2) \cap \cX \subseteq N(a_1) \cap \cX$. Suppose wlog.\ that the former holds. Together with (\ref{proof:hamcyc:properties:2}) and (\ref{proof:hamcyc:properties:5}), we can conclude that $N(a_1) \cap T_v \subseteq N(a_2) \cap T_v$, a contradiction.
		
		\item[Case 2 ($v \in L_k$ for some $k$).] Again by (\ref{proof:hamcyc:properties:3}), every vertex in $U_t$ where $t < k$ has no neighbors in $T_v$. By the argument given in Case $1$, it cannot happen that $a_1$ and $a_2$ are both contained in $L_1 \cup \cdots \cup L_{k-1} \cup (L_k \cap S_v)$. Hence we can (wlog.) assume that $a_1 \in U_k$, i.e.\ $a_1 = Y_k$ or $a_1 = Z_k$. Since $Y_k$ and $Z_k$ are twins by (\ref{proof:hamcyc:properties:3}), $a_2$ cannot be the vertex in $U_k \setminus \{a_1\}$ (if exists). We can conclude that $a_2 \in L_1 \cup \cdots \cup L_{k-1} \cup (L_k \cap S_v)$, in particular that $a_2 \in \cA$. By (\ref{proof:hamcyc:properties:3}), $N(a_1) = \cA_k$, so by (\ref{proof:hamcyc:properties:2}), every neighbor of $a_1$ is adjacent to $a_2$, a contradiction.
		
		\item[Case 3 ($v = X_k$ for some $k$).] Suppose wlog.\ that $b_1 = X_{j_1}$ and $b_2 = X_{j_2}$ where $j_1 < j_2 < k$. By (\ref{proof:hamcyc:properties:4}), $N(b_1) \subset N(b_2)$, so in particular $N(b_1) \cap S_v \subset N(b_2) \cap S_v$, a contradiction.
	\end{description}	
    We have shown that the linear branch decomposition (\ref{proof:hamcyc:branch:dec}) has mim-width $1$.
    \end{proof}

\section{Powers of graphs}\label{sec:graphclass}

	In this section we show that $k$-powers of graphs of tree-width $w-1$, path-width $w$, or clique-width $w$ all have mim-width at most $w$.
	This is somewhat surprising because this bound does not depend on $k$.
	We begin by proving the bound for graphs of bounded treewidth with the following lemma capturing the essential property used in the proof.
	Throughout the following, for a graph $G$ and a pair of vertices $v, w \in V(G)$, we denote by $\dist_G(v,w)$ the distance between $v$ and $w$ in $G$.
	
	\begin{lemma}\label{lem:power}
	Let $k$ and $w$ be positive integers and let $(A,B,C)$ be a vertex partition of  graph $G$ such that there are no edges between $A$ and $C$ and $B$ has size $w$.
	If $H$ is the $k$-power of $G$,
	then $\mimval_{H}(A\cup B)\le w$.
	\end{lemma}
	\begin{proof}
	Let $B\defeq \{b_1, b_2, \ldots, b_w\}$.
	It is clear that for $v\in A\cup B$ and $z\in C$, $\dist_G(v,z)\le k$ if and only if there exists $i\in \{1, 2, \ldots, w\}$ such that
	$\dist_G(v,b_i)+\dist_G(z,b_i)\le k$.
	
	Suppose for contradiction that $\mimval_{H}(A\cup B)> w$.
	Let $\{y_1z_1, y_2z_2, \ldots, y_tz_t\}$ be an induced matching of size $t \ge w+1$ in $H[A\cup B, C]$.
	There are distinct integers $t_1, t_2\in \{1, 2, \ldots, t\}$ and an integer $j\in \{1, 2, \ldots, w\}$ such that 
	\begin{align*}
		\dist_G(y_{t_1},b_j)+\dist_G(z_{t_1},b_j)\le k \text{ and } \dist_G(y_{t_2},b_j)+\dist_G(z_{t_2},b_j)\le k.
	\end{align*}
	Then we have either $\dist_G(y_{t_1},b_j)+\dist_G(z_{t_2},b_j)\le k$ or $\dist_G(y_{t_2},b_j)+\dist_G(z_{t_1},b_j)\le k$, 
	which contradicts the assumption that $y_{t_1}z_{t_2}$ and $y_{t_2}z_{t_1}$ are not edges in $H$.
	
	We conclude that $\mimval_{H}(A\cup B)\le w$.
	\end{proof}
	\begin{definition}
	A \emph{tree decomposition} of a graph $G$ is 
	a pair $(T,\cB)$ 
	consisting of a tree $T$
	and a family $\cB=\{B_t\}_{t\in V(T)}$ of sets $B_t\subseteq V(G)$,
	called \emph{bags},
	satisfying the following three conditions:
	\begin{enumerate}[(i)]
	\item $V(G)=\bigcup_{t\in V(T)}B_t$,
	\item for every edge $uv$ of $G$, there exists a node $t$ of $T$ such that $u,v\in B_t$, and\label{def:tree:decomposition:edge}
	\item for $t_1,t_2,t_3\in V(T)$, $B_{t_1}\cap B_{t_3}\subseteq B_{t_2}$ whenever $t_2$ is on the path from $t_1$ to $t_3$ in $T$.
	\end{enumerate}
	The \emph{width} of a tree decomposition $(T,\cB)$ is $\max\{ \abs{B_{t}}-1:t\in V(T)\}$.	
	The \emph{tree-width} of $G$ is the minimum width over all tree decompositions of $G$. 
	A tree decomposition $(T,\cB=\{B_t\}_{t\in V(T)})$ is a \emph{nice tree decomposition} with root node $r\in V(T)$ if $T$ is a rooted tree with root node $r$, and every node $t$ of $T$ is one of the following:
\begin{enumerate}[(1)]
  \item A \emph{leaf node}, i.e.\ $t$ is a leaf of $T$ and $B_t=\emptyset$.
  \item An \emph{introduce node}, i.e.\ $t$ has exactly one child $t'$ and $B_t=B_{t'}\cup \{v\}$ for some $v\in V(G)\setminus B_{t'}$.
  \item A \emph{forget node}, i.e.\ $t$ has exactly one child $t'$ and $B_t=B_{t'}\setminus \{v\}$ for some $v\in B_{t'}$.
  \item A \emph{join node}, i.e.\ $t$ has exactly two children $t_1$ and $t_2$, and $B_t=B_{t_1}=B_{t_2}$.
\end{enumerate}
\end{definition}
	
		\begin{theorem}\label{thm:treewidthpower}
			Let $k$ and $w$ be positive integers and $G$ be a graph that admits a nice tree decomposition of width $w$, all of whose join bags are of size at most $w$. Then the $k$-power of $G$ has mim-width at most $w$. Furthermore, given such a nice tree decomposition, we can output a branch decomposition of mim-width at most $w$ in polynomial time.
	\end{theorem}
	\begin{proof}
			Let $H$ be the $k$-power of $G$, and let $(T, \cB = \{B_t\}_{t\in V(T)})$ be a nice tree decomposition of $G$ of width $w$, all of whose join bags have size at most $w$, with root node $r$.
	We may assume that $B_r = \emptyset$ and subsequently that $r$ is a forget or a join node. (Otherwise, we add a path of forget nodes on top of $r$ and make the last node the new root of $T$.)
	
	We obtain a branch decomposition $(T', \decf)$ as follows:
	\begin{itemize}
		\item Let $T''$ be the tree obtained from $T$ by, for each forget node forgetting a vertex $v$, adding a leaf $\ell_v$, and assigning $\decf(v)\defeq \ell_v$.
		\item We obtain $T'$ from $T''$ by recursively smoothing degree $2$ nodes that are not the root node and degree $1$ nodes that are not assigned by $\decf$, where smoothing a node of degree $2$ is an operation of removing this node and adding an edge between its two neighbors.
	\end{itemize}
		Since for each vertex $v \in V(G)$, there is a unique forget node forgetting $v$ in $(T, \cB)$,
		the map $\decf$ constructed above is a bijection. 
		Furthermore, $r$ has degree $2$ in $T'$, since $r$ is a join or a forget node in $(T, \cB)$. Thus, $(T', \decf)$ is a rooted branch decomposition with root node $r$.
	
		We consider $\crossinggraph{t}$ for some $t \in V(T')$, the crossing graph w.r.t.\ $t$, and argue that $\mimval_H(V_t) \le w$.
		If $t$ is a leaf node, then $\mimval_H(V_t)\le 1$.
		Assume $t$ is an internal node, then $t$ also appears in $(T, \cB)$. We argue that we can find a set of at most $w$ vertices $S \subseteq \bar{V_t}$ such that $S$ separates $V_t$ from $\bar{V_t} \setminus S$ which by Lemma~\ref{lem:power} will yield the claim.
		
		We first observe that $S \defeq N(V_t) \cap \bar{V_t} \subseteq B_t$ and clearly $S$ separates $V_t$ from $\bar{V_t} \setminus S$. If $t$ is a forget node, then by definition $\card{B_t} \le w$ and hence $\card{S} \le w$. If $t$ is a join node, then by assumption $\card{B_t} \le w$ and hence $\card{S} \le w$. If $t$ is an introduce node introducing a vertex $u \in V(G)$, then $u$ cannot have any neighbor in $V_t$, since all vertices in $V_t$ have been forgotten below $t$. Hence, $S \subseteq B_t \setminus \{u\}$ and we can conclude that $\card{S} \le w$.
	\end{proof}
	It is well-known (see e.g.~\cite{Klo94}) that any tree decomposition can be transformed in polynomial time to a nice tree decomposition of the same width, hence the previous theorem implies	
	\begin{corollary}
		Let $k$ and $w$ be positive integers and let $G$ be a graph of tree-width $w - 1$ (path-width $w$). Then the $k$-power of $G$ has mim-width at most $w$ and given a tree decomposition (path decomposition) of $G$ of width $w - 1$ ($w$), one can compute a branch decomposition of mim-width $w$ in polynomial time.
	\end{corollary}
	The following notions are of importance in the field of phylogenetic studies, i.e. the reconstruction of ancestral relations in biology, see e.g. \cite{CS16}.
A graph $G$ is a leaf power if there exists a threshold $k$ and a tree
$T$, called a leaf root, whose leaf set is $V(G)$ such that $uv \in E$ if and only if the
distance between $u$ and $v$ in $T$
is at most $k$. Similarly, $G$ is called a min-leaf power if $uv \in E$ if and only if the
distance between $u$ and $v$ in $T$
is more than $k$. Thus, $G$ is a leaf power if an only if its complement is a min-leaf power.
	It is easy to see that trees admit nice tree decompositions all of whose join bags have size $1$ and	
	since every leaf power graph is an induced subgraph of a power of some tree, 
	it has mim-width at most $1$ by Theorem~\ref{thm:treewidthpower}.
	\begin{corollary}
		The leaf powers and min-leaf powers have mim-width at most $1$ and given a leaf root, we can compute in polynomial time a branch decomposition witnessing this.
	\end{corollary}	
	
	Next, we consider powers of graphs of bounded clique-width.
	A graph is \emph{$w$-labeled} if there is a labeling function $f:V(G)\rightarrow \{1, 2, \ldots, w\}$, and we call $f(v)$ the \emph{label} of $v$.  For a $w$-labeled graph $G$, we call the set of all
vertices with label $i$ the \emph{label class $i$} of $G$.
	The following can be thought as a generalization of Lemma~\ref{lem:power}.
	
	\begin{lemma}\label{lem:power2}
	Let $k$ and $w$ be positive integers and let $(A,B)$ be a vertex partition of  graph $G$ such that $G[A]$ is $w$-labeled and 
	for every pair of vertices $x, y$ in the same label class of $G[A]$, $x$ and $y$ have the same neighborhood in $B$.
	If $H$ is the $k$-power of $G$,
	then $\mimval_{H}(A)\le w$.
	\end{lemma}
	\begin{proof}
	Suppose for contradiction that $\mimval_{H}(A)> w$.
	Let $\{y_1z_1, y_2z_2, \ldots, y_tz_t\}$ be an induced matching of size at least $w+1$ in $H[A, B]$.
	For $i\in \{1, 2, \ldots, t\}$, there is a path $P_i$ of length at most $k$ from $y_i$ to $z_i$ in $G$. 
	Let $Q_i$ be the subpath of $P_i$ from $\bd(A)$ to $z_i$, and 
	$q_i$ be the end vertex of $P_i$ other than $z_i$, and 
	let $R_i$ be the subpath of $P_i$ from $y_i$ to $z_i$.
	Let $a_i$ be the length of $R_i$ and $b_i$ be the length of $Q_i$.
	By construction, $a_i+b_i\le k$.

	Since $t\ge w+1$, there are two integers $t_1, t_2\in \{1, 2, \ldots, t\}$ such that 
	$q_{t_1}$ and $q_{t_2}$ are contained in the same label class of $G[A]$.
	Thus, we have either $a_{t_1}+b_{t_2}\le k$ or $a_{t_2}+b_{t_1}\le k$.
	
	Assume $a_{t_1}+b_{t_2}\le k$. In this case, we show that the distance from $y_{t_1}$ to $z_{t_2}$ in $G$ is at most $k$, which contradicts the assumption that 
	$y_{t_1}z_{t_2}$ is not an edge of $H$.
	Note that two vertices in a label class of $G[A]$ have the same neighborhood in $B$.
	Thus, $q_{t_1}$ has a neighbor in the neighborhood of $q_{t_2}$ in $Q_{t_2}$.
	Therefore, $G[V(R_{t_1})\cup (V(Q_{t_2})\setminus \{q_{t_2}\})]$ contains a path of length at most $k$ from $y_{t_1}$ to $z_{t_2}$.
	Analogously we can show that if $a_{t_2}+b_{t_1}\le k$, then 
	$G[V(R_{t_2})\cup (V(Q_{t_1})\setminus \{q_{t_1}\})]$ contains a path of length at most $k$ from $y_{t_2}$ to $z_{t_1}$.
	But this is a contradiction and we conclude that $\mimval_{H}(A)\le w$.
	\end{proof}
	
\begin{definition}		
	The \emph{clique-width} of a graph $G$ is the minimum number of labels needed to construct $G$ using the following four operations:
\begin{enumerate}[(1)]
\item Creation of a new vertex $v$ with label $i$ (denoted by $i(v)$).
\item Disjoint union of two labeled graphs $G$ and $H$ (denoted by $G \oplus H$).
\item Joining by an edge each vertex with label $i$ to each vertex with label $j$ ($i\neq j$, denoted by $\eta_{i,j}$). 
\item Renaming label $i$ to $j$ (denoted by $\rho_{i\rightarrow j}$).
\end{enumerate}
\end{definition}
	A string of operations given in the previous definition is called a \emph{clique-width $k$-expression} or shortly a \emph{$k$-expression} if it uses at most $k$ distinct labels. We can represent this expression as a tree-structure and such trees are known as \emph{syntactic trees} associated with $k$-expressions. An easy observation is that for a node $t$ in a syntactic tree associated with a $k$-expression, and the vertex set $V_t$ consisting of vertices introduced in some descendants of $t$, 
$V_t$ is a $k$-labeled graph where two vertices in the same label class has the same neighborhood in $V(G)\setminus V_t$.

	\begin{theorem}
	Let $k$ and $w$ be positive integers and $G$ be a graph of clique-width $w$.
	Then the $k$-power of $G$ has mim-width at most $w$. Furthermore, 
	given a clique-width $w$-expression, we can output a branch decomposition of mim-width at most $w$ in polynomial time.
	\end{theorem}
	\begin{proof}
	Let $H$ be the $k$-power of $G$ and let $\phi$ be the given clique-width $k$-expression defining $G$, and  
	$T$ be the syntactic tree of $\phi$ with root node $r$. Let $r'$ be the node in $T$ of degree at least $3$ with minimum $\dist_T(r, r')$.
	Note that for every vertex $v$ of $G$, there is a node of $T$ creating $v$, see the first operation in the definition of clique-width. In the following, we call such a node an {\em introduce node}.
		We obtain a branch decomposition $(T', \decf)$ as follows:
	\begin{itemize}
	\item For each introduce node $\ell_v$ introducing a vertex $v$, we assign $\decf(v)\defeq \ell_v$.	
	\item We obtain $T'$ from $T$ as follows: If $r\neq r'$, we first remove all vertices in the path from $r'$ to $r$ in $T$ other than $r'$.
	We recursively smooth degree $2$ nodes that are not $r'$ (as in the proof of Theorem~\ref{thm:treewidthpower}).
	We fix $r'$ to be the root node of $T'$ as well.
	\end{itemize}

	Consider $\crossinggraph{t}$ for some $t \in V(T')$, the crossing graph w.r.t.\ $t$.
	If $t$ is a leaf node, then $\mimval_H(V_t)\le 1$.
	Assume $t$ is an internal node.
	Then $t$ also appears in $T$.
	We observe that $V_t$ is a $w$-labeled graph such that for any pair of vertices $x, y$ in the same label class, $x$ and $y$ have the same neighborhood in $V(G)\setminus V_t$.
	So we can apply Lemma~\ref{lem:power2} to conclude that we have $\mim_H(V_t)\le w$ which implies that $H$ has mim-width at most $w$.
	\end{proof}
	
\section{Conclusion}

We have shown that \textsc{Feedback Vertex Set} admits an $n^{\mathcal{O}(w)}$-time algorithm when given a branch decomposition of mim-width $w$. This provides a unified polynomial-time algorithm for {\sc Feedback Vertex Set} on known classes of bounded mim-width, and gives the first polynomial-time algorithms for \textsc{Circular Permutation} and \textsc{Circular $k$-Trapezoid} graphs for fixed $k$.

Somewhat surprisingly, we prove that powers of graphs of bounded tree-width or clique-width have bounded mim-width.
Heggernes et al.~\cite{Heggernes2016} showed that the clique-width of the $k$-power of a path of length $k(k+1)$ is exactly $k$.
This also shows that the expressive power of mim-width is much stronger than clique-width, since all powers of paths have mim-width just $1$.
As a special case, we show that {\sc Leaf Power} graphs have mim-width $1$.
We believe the notion of mim-width can be of benefit to the study of {\sc Leaf Power} graphs.

We conclude with repeating an open problem regarding algorithms for computing mim-width.
The problem of computing the mim-width of general graphs was shown to be 
$\W[1]$-hard, not in $\APX$ unless $\NP = \ZPP$~\cite{SV16} and no algorithm for computing the mim-width of 
a graph in $\XP$ time is known. We therefore repeat an open question from~\cite{SV16}:
\begin{question*}[see also \cite{SV16}, cf.~\cite{VatshelleThesis}]
	Is there an $\XP$ algorithm approximating mim-width $w$ by some function $f(w)$ and returning a decomposition?
\end{question*}
We remark that it is a big open problem whether {\sc Leaf Power} graphs can be recognized in polynomial time~\cite{Bra10,CS16,Lafond17,NR16}. 
A positive answer to our open problem may be used to design such a recognition algorithm 
using branch decompositions of bounded mim-width.

\bibliography{fvsmim}
\end{document}